\documentclass[%
 reprint,
 amsmath,amssymb,
 aps,
]{revtex4-2}

\usepackage{graphicx}% Include figure files
\usepackage{dcolumn}% Align table columns on decimal point
\usepackage{bm}% bold math
\usepackage[font=small,labelfont=small]{caption}
\usepackage[colorlinks=true, linkcolor=blue, urlcolor=blue, citecolor=blue]{hyperref}

\usepackage{graphicx}
\usepackage{booktabs}
\usepackage{amsthm}
\usepackage{amssymb} 
\usepackage{mathrsfs}
\usepackage{amsmath}
\usepackage{cleveref}
\usepackage{braket}

\usepackage{algpseudocode}
\usepackage{algorithm}
\usepackage{subcaption}

\newtheorem{thm}{Theorem}

\newtheorem{lem}{Lemma}[section]

\begin{document}

\preprint{APS/123-QED}

%\title{Adiabatic-Passage-Based Parameter Setting for Quantum Approximate Optimization Algorithm} % on 3-SAT Problem
% Origin of Efficacy in QAOA for Random k-SAT: Sublinear Optimization via Hierarchical Adiabatic-Manifold Parameterization (for QST)
\title{Mechanism of Efficacy in QAOA for Random \textit{k}-SAT: \\ From Adiabatic Manifold to Sublinear Parameter Optimization}  % Manifold,Trajectory,Profile-

\author{Mingyou Wu}
\email{wumy@seu.edu.cn}
\altaffiliation{School of Computer Science and Engineering, Southeast University, Nanjing 211189, China}%Lines break automatically or can be forced with \\
%\affiliation{%

\author{Hanwu Chen}
\email{hw\_chen@seu.edu.cn}
\affiliation{School of Computer Science and Engineering, Southeast University, Nanjing {\rm 211189}, China}
\affiliation{Key Laboratory of Computer Network and Information Integration (Southeast University), Ministry of Education,  Nanjing {\rm 211189}, China}
%\author{Zhihao Liu}
%\email{lzh@seu.edu.cn}
%\affiliation{School of Computer Science and Engineering, Southeast University, Nanjing {\rm 211189}, China}
%\affiliation{Key Laboratory of Computer Network and Information Integration (Southeast University), Ministry of Education,  Nanjing {\rm 211189}, China}


%\affiliation{
%\affiliation{
%\affiliation{%


\date{\today}% It is always \today, today,

\begin{abstract}
	% the main work
	The Quantum Approximate Optimization Algorithm (QAOA) is a leading candidate for demonstrating quantum advantage on near-term devices, yet the physical origins of its efficacy remain poorly understood. In this work, we study QAOA for random \textit{k}-SAT problems within a universal-mixer \textit{k}-local search framework, establishing a formal correspondence between adiabatic state transfer and the QAOA ansatz. This correspondence yields a rigorous performance guarantee for random instances with clause density $m=\mathcal{O}(n^{1+\epsilon})$ and circuit depth $\Theta(n^2)$.
	
	We further investigate the NISQ regime with shallow circuits of depth $p=\mathcal{O}(n)$. Surprisingly, the optimal parameters do not become stochastic under depth compression, but instead remain confined to a structured low-dimensional region, which we identify as a \textit{smooth adiabatic manifold}. Numerical evidence indicates that this manifold persists across different circuit depths and arises from the variational suppression of adiabatic leakage. Based on this structure, we propose the smooth adiabatic-manifold parameterization (SAMP) strategy, transforming parameter optimization from an unstructured high-dimensional search into a guided refinement process. Numerical experiments on random 3-SAT instances show that SAMP achieves sublinear optimization scaling with circuit depth while providing robust zero-cost initialization for deep circuits.
%\keywords{quantum computing, quantum approximate optimization algorithm, quantum adiabatic computation, combinatorial optimization, 3-satisfiability problem}
\end{abstract}

\maketitle

%\tableofcontents

\section{Introduction}\label{sec1}
% 首先就说量子计算基础，目前的一个NISQ环境，以及对应的变分算法
Quantum computation offers a fundamentally new paradigm for information processing, with provable speedups over classical approaches for certain problems, exemplified by Shor’s algorithm for integer factorization~\cite{Shor1994} and Grover’s algorithm for unstructured search~\cite{Grover1997}. Despite this theoretical promise, the realization of large-scale, fault-tolerant quantum computers remains a long-term challenge. In the current Noisy Intermediate-Scale Quantum (NISQ) era~\cite{Preskill2018}, limitations in qubit coherence and gate fidelity preclude the direct implementation of deep, error-corrected circuits.

% 介绍对应的
These hardware limitations have led to the development of hybrid quantum-classical approaches, most notably variational quantum algorithms (VQAs)~\cite{Cerezo2021}. By combining shallow parameterized quantum circuits with classical optimization, VQAs aim to extract computational advantages while remaining compatible with near-term devices. Representative examples include the Variational Quantum Eigensolver (VQE)~\cite{Peruzzo2014} and Quantum Neural Networks (QNNs)~\cite{Farhi2018QNN}.

% 引入QAOA
Among these, the Quantum Approximate Optimization Algorithm (QAOA)~\cite{Farhi2014} has emerged as a central framework for combinatorial optimization. Constructed from an alternating application of a problem Hamiltonian and a mixer Hamiltonian, QAOA prepares a parameterized quantum state intended to approximate the ground state of complex objective functions. Beyond its practical appeal, the structure of QAOA naturally connects it to broader models of quantum computation, enabling universal quantum computation under suitable constructions~\cite{Lloyd2018, Hadfield2019}, while in certain limits reproducing optimal quantum search behavior for unstructured problems~\cite{Jiang2017, Morales2020}.

% 提出gap
Despite this progress, a significant gap remains between the empirical performance of QAOA and its theoretical foundation. While QAOA often exhibits remarkable success in numerical experiments---particularly on structured and random instances~\cite{Harrigan2021, Zhou2020, Wurtz2021Setting}---it lacks a unified theoretical framework that explains its efficacy or provides performance guarantees beyond limited settings. This deficiency is especially pronounced for NP-hard problems, where both the structure of the optimization landscape and the scaling behavior of the algorithm remain poorly understood.

% 具体说明这个gap附近的研究文献
Existing analytical approaches provide only partial insight into this landscape, and are largely confined to extreme regimes of circuit depth. On one hand, complexity-theoretic results establish fundamental limitations in the shallow-depth limit, demonstrating that such circuits cannot surpass certain classical approximation barriers~\cite{Farhi2020Model, Bravyi2020}. On the other hand, connections to continuous-time or adiabatic quantum evolution suggest that achieving high-quality solutions for hard instances may require circuit depths that scale substantially with the problem size~\cite{Lykov2023, Basso2021}. Taken together, these limitations leave a theoretical gap at moderate, polynomially-scaled depths, where rigorous performance guarantees remain largely absent for general optimization problems.

%instance-specific analyses---for example, studies of QAOA on highly symmetric settings such as Max-Cut on regular graphs~\cite{Wang2018, Wang2022_PRX_Quantum}---provide valuable analytical characterizations, but are inherently restricted to special cases and do not readily extend to hard random instances. 

% 说明为了弥补这个gap，我们做了什么。
In this work, we bridge this gap by establishing a mechanistic foundation for the efficacy of QAOA, specifically focusing on random $k$-SAT. As a canonical NP-complete problem, $k$-SAT captures the intrinsic hardness of combinatorial optimization, while its random ensemble provides a natural setting for studying average-case complexity~\cite{Levin1986} and typical-instance behavior. To facilitate a rigorous analysis within the QAOA framework, we introduce a restricted optimization variant termed \textit{max-$k$-SSAT}. This formulation preserves the core computational structure of $k$-SAT while providing an optimization objective amenable to analytical treatment. Crucially, the efficient solvability of \textit{max-$k$-SSAT} directly implies a solution to the standard $k$-SAT decision problem with the same asymptotic complexity. A detailed formulation and proof are provided in Sec.~\ref{subsec_reduction}.

% 接下来就说明我们具体做了什么
Our analysis begins by establishing a connection between QAOA and the recently proposed $k$-local quantum search algorithm~\cite{mine1}. We generalize the mixer operator in $k$-local quantum search, leading to a formulation under which the QAOA ansatz emerges as a special case. Within this unified formulation, the QAOA evolution can be interpreted as a variational adaptation of an adiabatic $k$-local quantum search process. Building on this connection, we rigorously establish that such a correspondence is not merely structural, but also preserves the essential performance characteristics of the underlying search dynamics. Concretely, for any fixed $\epsilon > 0$, when the circuit depth scales as $p = \Theta(n^2)$ and the clause density satisfies $m = \Theta(n^{1+\epsilon})$, QAOA solves typical instances of random max-$k$-SSAT with probability $1 - \mathcal{O}(\mathrm{erfc}(n^{\delta/2}))$.
%Moreover, when $m = \Theta(n^{2+\epsilon})$, the problem becomes solvable in polynomial time on average.

% 这下面再引入参数设定的问题。
% 首先，说明这个结果的意义，这意味这当QAOA以 n^2的深度是，绝热路径是一条天然的优化路径通向最优解。这提供了远超给予Grover的有效性保证。
% --- 从理论保证向参数设定的过渡 ---
These results indicate that, already at quadratic depth, QAOA can reliably solve typical instances of random $k$-SAT with superlinear clause densities. This provides a concrete theoretical basis for its efficacy that far exceeds the guarantees derived from unstructured Grover-type search. More importantly, this correspondence reveals that the optimal variational trajectory is not an arbitrary path in the $2p$-dimensional parameter space, but is fundamentally rooted in a discretized adiabatic passage. In this regime, the adiabatic trajectory serves as a ``natural'' optimization path that guides the system toward the global minimum.

%而当深度降低时，反正就是绝热泄露那一系列问题吗，但是我们同样发现了，压缩深度的绝热k-local搜索，对于k-local搜索问题有效，而这个，就很像是标准型，其他的优化后的参数就收敛于这个。这与之前的大量研究吻合，就可以引入之前的研究了。

% 下面在基于这一点，说我们方法单纯以这个为基准，而不是盲人摸象，直接深入使用特性，得到了SAMP，这个之后再说。

% --- 讨论深度压缩、绝热泄露与变分补偿机制 ---
However, practical implementation on NISQ devices necessitates a significant compression of circuit depth, often down to linear or even sublinear $p$. In such restricted regimes, the strict adiabatic condition is inevitably violated, leading to non-adiabatic transitions and what is known as ``adiabatic leakage''. We contend that the true power of the variational loop in QAOA lies in its ability to ``heal'' these transitions—acting as a variational shortcut to adiabaticity. This perspective is supported by a growing body of evidence suggesting that optimal QAOA parameters do not behave like the stochastic initializations typically plagued by \textit{barren plateaus}~\cite{McClean2018, Wang2021}, but instead exhibit a high degree of geometric continuity.

% 讨论参数聚集的现状
Previous empirical studies have observed that for several combinatorial problems, optimal parameters tend to cluster around smooth, predictable trajectories~\cite{Zhou2020, Brandao2018}. This has inspired a variety of heuristic strategies to mitigate the high-dimensional optimization challenge, including interpolation-based parameter transfer schemes~\cite{Zhou2020}, adiabatic-inspired schedules~\cite{Sack2021, Sud2024, Wurtz2021Setting} and warm-starting approaches that extrapolate parameters across instances or system sizes~\cite{Egger2021Warm, Tate2023}. Complementing these heuristics, theoretical investigations on structured problems—most notably Max-Cut on regular graphs—have identified a phenomenon known as \textit{parameter concentration}~\cite{Wang2018, Wang2022, Brandao2018}. 

% --- 提出 Adiabatic Manifold 的观点 ---
Together, these results strongly suggest that the QAOA parameter landscape is not an unstructured variational space, but is instead governed by an underlying low-dimensional global structure. In this work, we provide a rigorous basis for this observation by analytically characterizing the parameter concentration regions for \textit{max-k-SSAT}. Specifically, we demonstrate that at a circuit depth of $p = \Omega(n^2)$, the optimal parameters converge toward a well-defined adiabatic trajectory with high probability. Crucially, our analysis extends beyond this theoretical limit to the NISQ-relevant regime of linear circuit depth. We show that as the depth is compressed from $\Omega(n^2)$ to $\mathcal{O}(n)$, the core geometric features of the parameter clustering persist, manifesting as a robust, low-frequency envelope. We formalize this persistent structure as the \textit{adiabatic manifold}---a smooth geometric subspace that preserves the essential characteristics of the optimal adiabatic passage even under significant depth compression. 

% --- 提出我们的参数设定方法 (SAMP) ---
By identifying the adiabatic manifold as the true locus of algorithmic efficacy, we move beyond heuristic parameter-tuning to propose a principled framework that exploits this geometric invariance: the \textit{smooth adiabatic-manifold parameterization} (SAMP). The core of SAMP lies in a reparameterization of the variational space designed to align with the manifold's curvature, thereby minimizing the inter-layer deviations between parameters. This transformation drastically enhances the exploitability of the passage's continuity, enabling efficient parameter reuse and interpolation across varying depths. To navigate this manifold efficiently, we implement a hierarchical refinement strategy: starting from a single degree of freedom to capture the global trend, we systematically double the degrees of freedom to capture finer diabatic corrections. This approach refocuses the classical optimizer on the most physically relevant regions, ultimately achieving a substantial reduction in optimization overhead and demonstrating the potential for sublinear scaling in parameter-setting costs.

The remainder of this paper is organized as follows. In Sec.~\ref{sec_reduction}, we introduce the formal reduction from $k$-SAT to the optimization variant max-\textit{k}-SSAT and perform a statistical characterization of the problem Hamiltonian. In Sec.~\ref{sec_theory}, we present the theoretical foundation of QAOA's efficacy, establishing its connection to $k$-local quantum search and proving the existence of an optimal adiabatic trajectory at the $\Omega(n^2)$ depth limit. Sec.~\ref{sec_SAMP} details the formulation of SAMP method, including the reparameterization of the variational space and the hierarchical refinement strategy. Numerical results on random 3-SAT instances, demonstrating the sublinear scaling of optimization costs, are presented in Sec.~\ref{sec_results}. Finally, Sec.~\ref{sec_conclusion} concludes the paper with a discussion on the transferability of the SAMP framework to broader classes of combinatorial problems and its potential to catalyze practical quantum advantage on near-term NISQ hardware, followed by an outlook on future research directions.

\section{Problem reduction and statistical analysis}\label{sec_reduction}
% 这一部分将k-SAT规约到max-k-SSAT，并对max-k-SSAT的问题哈密顿量进行统计分析

% 先来个总起
In this section, we establish the mathematical framework for evaluating QAOA on the typical-case instances of $k$-SAT. To bridge the gap between the Boolean decision problem and a variational quantum optimization task, we introduce a restricted optimization variant termed max-$k$-SSAT. Crucially, max-$k$-SSAT preserves the fundamental hardness of $k$-SAT, as its efficient solvability directly implies a solution to the standard decision problem. Our objective is to characterize the statistical properties of the problem Hamiltonian under ensembles that represent typical satisfiable instances, thereby providing a rigorous foundation for the efficacy guarantees presented in later sections.
% 前者作为最终证明目标，后者作为理论分析的入手点。然后就以后者为基础，讨论了特性。
\subsection{Problem reduction} 
\label{subsec_reduction}% : from \textit{k}-SAT to max-\textit{k}-SSAT

The \textit{Boolean satisfiability} (SAT) problem asks whether a given Boolean formula admits an interpretation that evaluates to true. In the context of $k$-SAT, the formula is expressed in conjunctive normal form (CNF), consisting of a conjunction of $m$ clauses, where each clause is a disjunction of exactly $k$ literals. While $k$-SAT is the prototypical NP-complete problem for $k \ge 3$~\cite{cook1971complexity, Karp1972}, this classification primarily characterizes its worst-case intractability. It does not imply that the problem is universally challenging; indeed, much of the practical interest lies in its average-case complexity under specific distributions, which often reveals a rich structure of typical-case hardness~\cite{Levin1986}. The details of average-case complexity theory is displayed in Appendix~\ref{apsubsec_average_case}. 

To address $k$-SAT within a variational quantum framework, it is natural to consider its optimization counterpart, max-\textit{k}-SAT. The goal here is to identify a variable assignment $x \in \{0,1\}^n$ that satisfies as many clauses as possible, thereby maximizing the objective function:
\begin{equation}
	f(x) = \sum_{\alpha=1}^m f_\alpha(x),
\end{equation}
where $\alpha$ indexes the clauses and $f_\alpha(x)$ is the characteristic function that yields $1$ if clause $\alpha$ is satisfied and $0$ otherwise. 

However, performing an unconstrained optimization over the entire instance space presents significant analytical and computational hurdles. While $k$-SAT and max-$k$-SAT are polynomially reducible to each other, the practical cost of these reductions is asymmetric. Solving $k$-SAT via a max-$k$-SAT oracle requires only a single query; conversely, solving max-$k$-SAT using a $k$-SAT oracle typically necessitates multiple iterative queries. 

Given that our goal is to establish a mechanistic foundation for algorithm design rather than addressing the most general optimization case, we shift our focus to a restricted variant termed max-\textit{k}-SSAT. This problem is confined to the subspace of satisfiable instances and aims to identify a satisfying assignment $t$ that effectively reaching the global maximum $f(t)=m$. The following lemma establishes that this restriction does not sacrifice the fundamental computational hardness of the $k$-SAT decision problem.

\begin{lem}\label{lem_problem_reduction}
	If max-\textit{k}-SSAT is solvable within a time complexity of $\mathcal{O}(f(n))$, then the $k$-SAT decision problem is also solvable within the same complexity magnitude.
\end{lem}
\begin{proof}
	Let $U$ denote the set of all $k$-SAT instances on $n$ variables, which can be partitioned into two disjoint subsets: $U_s$ (satisfiable instances) and $U_u$ (unsatisfiable instances). Suppose there exists an algorithm $\mathcal{A}$ that, for any instance $I \in U_s$, returns a satisfying assignment within time $\mathcal{O}(f(n))$.
	
	We construct a decision algorithm $\mathcal{A}'$ for $k$-SAT as follows. Given an arbitrary instance $I \in U$, we execute $\mathcal{A}$ for at most $\mathcal{O}(f(n))$ steps. If $\mathcal{A}$ terminates and outputs an assignment $x$ within this bound, we verify whether $x$ satisfies $I$. If the verification is successful, $\mathcal{A}'$ concludes that $I \in U_s$. If $\mathcal{A}$ fails to provide an assignment within the time limit, or if the output fails verification, $\mathcal{A}'$ concludes that $I \in U_u$. 
	
	The correctness of $\mathcal{A}'$ follows directly from the properties of $U_s$ and $U_u$. If $I \in U_s$, $\mathcal{A}$ is guaranteed by assumption to return a satisfying assignment within $\mathcal{O}(f(n))$, which will pass verification. Conversely, if $I \in U_u$, no assignment can satisfy $I$; thus, any output from $\mathcal{A}$ (or its failure to terminate) will fail to satisfy the formula, leading $\mathcal{A}'$ to correctly output that the instance is unsatisfiable. Since the verification step is efficient, the total complexity of $\mathcal{A}'$ remains $\mathcal{O}(f(n))$.
\end{proof}

Based on this, the remainder of our theoretical investigation and algorithmic development will focus on random max-\textit{k}-SSAT. To evaluate its performance in typical-case scenarios, we now proceed to characterize the statistical properties of the problem instances.

\subsection{Random models and statistical analysis}
\label{subsec_statistical_analysis}

To evaluate the performance of QAOA on typical instances, we utilize the framework of average-case computational complexity theory~\cite{Levin1986}; the foundational details of this theory are presented in Appendix~\ref{apsubsec_average_case}. For our analysis, we employ concrete random ensembles, starting with the standard random $k$-SAT model $F(n,m,k)$. In this ensemble, an instance on $n$ variables is generated by uniformly and independently selecting $m$ clauses (with replacement) from the set of all $2^k C_n^k$ possible clauses~\cite{Achlioptas2006}, where $C_n^k$ denotes the binomial coefficient $\binom{n}{k}$. To align with the max-$k$-SSAT formulation, we consider two specific derivative models.

The first is the $F_s(n,m,k)$ model, a satisfiable-restricted ensemble derived from $F(n,m,k)$ by conditioning on satisfiability (i.e., through rejection sampling). This model represents the typical instances encountered in practical optimization and serves as our primary target for evaluating algorithmic efficacy. The second is the $F_f(n,m,k)$ model, a fixed-assignment proxy ensemble where a specific satisfying interpretation $t_0$ is fixed a priori, and clauses are sampled only from the $2^k-1$ configurations satisfied by $t_0$.

The $F_f$ model provides a mathematically tractable starting point for our analysis, as its constructive nature allows for a rigorous derivation of the Hamiltonian's spectral properties. While $F_s$ remains our ultimate target for performance guarantees, we leverage the $F_f$ model to extract the instance-independent statistical characteristics that underpin our parameter-setting strategy. In the following, we proceed to characterize the problem Hamiltonian under these random ensembles to establish a basis for the subsequent analysis.

Based on the $F_f(n,m,k)$ model, we analyze the statistical properties of the problem Hamiltonian $H_C$. In this ensemble, each clause $\alpha$ is sampled independently from the set of clauses satisfied by a fixed assignment $t$. Consequently, the Hamiltonian contribution of each clause, $H_\alpha$, can be treated as a random diagonal operator. Its eigenvalue $E_{\alpha,x}$ at basis state $\left| x \right\rangle$ is a Bernoulli random variable, where $E_{\alpha,x} = 1$ if $x$ satisfies $\alpha$, and $0$ otherwise. 

By combinatorial counting of the clause space, the mean and variance of $E_{\alpha,x}$ are derived as:
\begin{align}
	& {{\mu }_{k,x}}=\frac{{{2}^{k}}-2}{{{2}^{k}}-1}+\frac{C_{l}^{k}}{\left( {{2}^{k}}-1 \right)C_{n}^{k}}\le 1, \nonumber\\
	&  \sigma _{k,x}^{2}={{\left( 1-{{\mu }_{k,x}} \right)}^{2}}{{\mu }_{k,x}}+\frac{\mu _{k,x}^{2}\left( C_{n}^{k}-C_{l}^{k} \right)}{\left( {{2}^{k}}-1 \right)C_{n}^{k}}\le \frac{1}{{{2}^{k}}-1},
\end{align}
where $l = n - d_H(x,t)$ represents the number of bits where $x$ and $t$ agree ($d_H$ being the Hamming distance). 

Since $H_C = \sum_\alpha H_\alpha$, the total eigenvalues ${\mathcal{E}}_{k,x} = \sum_\alpha E_{\alpha, x}$ are the sum of $m$ i.i.d. variables. According to the Central Limit Theorem, for large $m$, the scaled eigenvalue ${\mathcal{E}}_{k,x}/m$ converges to a normal distribution:
\begin{equation}\label{eq_normal_distribution}
	\sqrt{m}\left( \frac{1}{m}{\mathcal{E}}_{k,x}-{{\mu }_{k,x}} \right)\sim N(0,\sigma _{k,x}^{2}).
\end{equation}

This statistical convergence reveals a profound property: as the clause density increases, the scaled operator $H_C/m$ concentrates around a deterministic ``standard form.'' To align this spectrum with the $k$-local search framework, we define the normalized problem Hamiltonian as:
\begin{equation}\label{eq_bar_HC}
	\bar{H}_C = \frac{(2^k-1)H_C}{m} - (2^k-2).
\end{equation}
The expectation value over the $F_f$ ensemble is thus
\begin{equation}\label{eq_normalized_HC}
	\mathbb{E}\!\left[\langle x \lvert\bar{H}_C\lvert x \rangle\right] = \frac{C_{l}^{k}}{C_{n}^{k}}.
\end{equation}
Remarkably, this expectation value exactly matches the Hamiltonian $H_k$ of the structured $k$-local search problem analyzed in Ref.~\cite{mine1}, which is defined by the objective function 
\begin{equation}\label{eq_k_local_search}
	f_k(x) = \frac{C_l^k}{C_n^k}.
\end{equation}

This correspondence suggests that a random max-$k$-SSAT instance can be viewed as a $k$-local search problem with ``missing'' clauses. This perspective provides a theoretical foothold for efficiently solving random max-$k$-SSAT instances under certain conditions. Furthermore, when employing variational methods, the optimal variational trajectory for a random instance is, on average, governed by the same adiabatic manifold as its structured counterpart. These two insights constitute the central points of our paper, which will be elaborated in detail in the following sections.

\section{Adiabatic origin of QAOA efficacy}\label{sec_theory}

This section aims to establish the theoretical foundation for the efficiency of QAOA applied to the random max-$k$-SSAT. Specifically, we first introduce the $k$-local quantum search algorithm recently proposed in Ref.~\cite{mine1}, which exhibits efficiency for max-$k$-SSAT instances with high clause density. However, its quantum gate complexity is substantial, scaling as $O(n^k + m)$. To reduce this complexity, we propose the universal-mixer variant of the $k$-local quantum search, and in Appendix~\ref{apsec_proof_universal_mixer}, we provide proofs that it preserves the efficiency of the original construction. Building on this foundation, we establish a connection between the adiabatic $k$-local quantum search and QAOA, allowing the latter to inherit the efficiency guarantees demonstrated for the former.

\subsection{\textit{k}-local quantum search}% and its adiabatic variant
\label{subsec_k_local_quantum_search}

To establish the theoretical foundation for our approach, we briefly review the \textit{k}-local quantum search algorithm introduced in Ref.~\cite{mine1}. This algorithm is designed to solve the \textit{k}-local search problem defined by the objective function $f_k(x)$ in Eq.~(\ref{eq_k_local_search}), and evolves according to
\begin{equation}\label{eq_k_local_circuit}
	|\psi\rangle = \left( e^{-i\theta H_{B,k}} e^{-i\theta H_k} \right)^p |+\rangle^{\otimes n},
\end{equation}
where $H_k$ encodes the objective function via $H_k |x\rangle = f_k(x) |x\rangle$, and $H_{B,k}$ is a $k$-local mixer Hamiltonian matched to the locality of $H_k$. Specifically, it is defined as $H_{B,k} = H^{\otimes n} H_{k,0} H^{\otimes n}$, where $H_{k,0}$ denotes the $k$-local Hamiltonian corresponding to the reference target $t=0$.

% 具体多少成功概率要讲清楚，这样阅读正文不需要去access appendices
For small constant $k$, this algorithm achieves efficient amplitude amplification with $p = \mathcal{O}(n^2)$ for an appropriate choice of $\theta = \Theta(n^{-1})$~\cite{mine1}. Moreover, this framework extends to random max-$k$-SSAT instances by replacing $H_k$ with the normalized Hamiltonian $\bar{H}_C$. Specifically, for any constant $\epsilon>0$ and sufficiently large $n$, the algorithm remains efficient with high probability $1-\mathcal{O}(\mathrm{erfc}(n^{\delta/2}))$ provided that the clause density satisfies $m = \Theta(n^{2+\epsilon+\delta})$ (see Ref.~\cite{mine1} and Appendix~\ref{apsubsubsec_circuit_k_local_QS} for details).

To improve robustness against deviations in the problem Hamiltonian, an adiabatic variant was introduced based on the interpolation
\begin{equation}\label{eq_Hamiltonian_AQS}
	\bar{H}_k(s) = sH_k + (1-s)H_{B,k}.
\end{equation}
Under a linear schedule $s = t/T$, the evolution remains efficient with total runtime $T = \mathcal{O}(n^2)$~\cite{mine1}. When applied to random max-$k$-SSAT instances, replacing $H_k$ by $\bar{H}_C$ yields the interpolation 
\begin{equation}
	H_k(s) = s\bar{H}_C + (1-s)H_{B,k}.
\end{equation} 
In this setting, the system is initialized in the uniform superposition $|+\rangle^{\otimes n}$, which is the highest-energy eigenstate of $H_{B,k}$, and the adiabatic evolution reduces the required clause density to $m = \Theta(n^{1+\epsilon+\delta})$, while achieving the same asymptotic success probability (see also Appendix~\ref{apsubsubsec_circuit_adiabatic_k_local_QS} for details).

Conceptually, the \textit{k}-local quantum search can be viewed as a structured generalization of unstructured quantum search. In particular, when $k=n$ and $\theta=\pi$, Eq.~(\ref{eq_k_local_circuit}) reduces to Grover's search in its circuit form, while the corresponding adiabatic formulation recovers the adiabatic Grover algorithm~\cite{Roland2002}. From a complementary perspective, the mixer Hamiltonian $H_{B,k}$ admits a natural interpretation as a $k$-local generalization of the standard transverse-field operator. Notably, when $k=1$, $H_{B,1}$ reduces to the conventional mixer $H_B = \sum_j \sigma_j^x$ (up to normalization and a global phase). 

\subsection{Universal-mixer \textit{k}-local quantum search}
\label{subsec_universal_mixer_QS}

The adiabatic $k$-local quantum search algorithm introduced in the previous section demonstrates impressive efficiency in theory. However, its practical implementation in the NISQ regime is severely limited by circuit depth. Each iteration of the algorithm incurs a cost of $O(m + n^k)$, and the number of iterations required scales as $\Theta(n^2)$, making direct implementation infeasible for realistic system sizes. This limitation motivates the exploration of alternative approaches.

For small constant $k$, the difference between the original $k$-local mixer $H_{B,k}$ and the simpler $1$-local mixer $H_{B,1}$ is effectively bounded in operator norm. This observation motivates a natural generalization of the algorithm, in which the $k$-local mixer is replaced by the $1$-local mixer, yielding a universal-mixer variant:
\begin{equation}\label{eq_universal_mixer_circuit}
	\ket{\psi} = \left( e^{-i\theta H_{B,1}} e^{-i\theta H_k} \right)^p \ket{+}^{\otimes n}.
\end{equation}
By removing the explicit locality matching in the mixer, this modification brings the circuit structurally closer to the standard QAOA ansatz, which employs a fixed, problem-independent mixer. 

% 说明应用方式，以及性能的具体证明
When applied to the max-$k$-SSAT problem, the $k$-local Hamiltonian $H_k$ can be replaced by the corresponding problem Hamiltonian $\bar{H}_C$. Under this setting, the quantum gate complexity per iteration is reduced to $O(m + n)$. Regarding algorithmic performance, Appendix~\ref{apsubsec_proof_search_on_search} and \ref{apsubsec_proof_search_on_SAT} provide detailed proofs that this universal-mixer variant preserves the fundamental structure of the original $k$-local quantum search algorithm. In particular, for the max-$k$-SSAT problem, the algorithm remains efficient with high probability $1 - \mathcal{O}(\mathrm{erfc}(n^{\delta/2}))$, provided that the clause density scales as $m = \Theta(n^{2+\epsilon+\delta})$.

% 接着就是引入绝热部分的讨论
% 这里就介绍绝热的基础知识，因为没有单开，之后也要用到
To further relax the clause-density requirement, we introduce an adiabatic variant of the universal-mixer construction. The time-dependent system Hamiltonian is defined as
\begin{equation}
	\bar{H}(s) = s H_k + (1-s) H_{B,1}.
\end{equation}
For a successful adiabatic evolution, the total runtime $T$ must satisfy
\begin{equation}\label{eq_adiabatic_condition}
	T \gg \frac{\varepsilon_0}{g_0^2},
\end{equation}
where $g_0 = \min_s g_k(s)$ is the minimum spectral gap along the interpolation. The gap $g_k(s)$ is the difference between the largest and second-largest eigenvalues of $\mathcal{H}_k(s)$, reflecting the maximization setting of the objective function $f_k(x)$. The quantity $\varepsilon_0$ characterizes the maximum transition amplitude induced by the variation of the Hamiltonian along the interpolation, and is given by
\begin{equation}\label{eq_T_varepsilon}
	\varepsilon_0 = \max_s \left| \left\langle \psi_0(s) \middle| \frac{d}{ds} \mathcal{H}_k(s) \middle| \psi_1(s) \right\rangle \right|,
\end{equation}
where $|\psi_0(s)\rangle$ and $|\psi_1(s)\rangle$ denote the instantaneous eigenstates associated with the largest and second-largest eigenvalues, respectively.

Under a linear schedule $s = t/T$, the adiabatic condition implies that the runtime scales as $T = \mathcal{O}(g_0^{-2})$. For the \textit{k}-local search problem, we show that the minimum spectral gap satisfies $g_0 = \Theta(n^{-1})$, which leads to an overall runtime $T = \mathcal{O}(n^2)$. This analysis extends to random max-$k$-SSAT instances by replacing $H_k$ with the normalized problem Hamiltonian $\bar{H}_C$, yielding
\begin{equation}
	H(s) = s \bar{H}_C + (1-s) H_{B,1}.
\end{equation}
As shown in Appendix~\ref{apsubsec_proof_adiabatic_on_SAT}, for any constant $\epsilon>0$ and sufficiently large $n$, the adiabatic evolution remains effective at clause density $m = \Theta(n^{1+\epsilon+\delta})$, achieving a success probability of $1-\mathcal{O}(\mathrm{erfc}(n^{\delta/2}))$.

\subsection{Connecting QAOA with adiabatic \textit{k}-local quantum search}
\label{subsec_connecting_qaoa}

To implement the adiabatic $k$-local quantum search within a circuit model, we discretize the dynamics using a Trotter--Suzuki decomposition. By adopting a linear schedule $s = t/T$ and dividing the total runtime $T$ into $p$ steps of duration $\Delta t = T/p$, the resulting state can be written as
\begin{equation}\label{eq_trotter_adiabatic_search}
	\ket{\psi} = \prod_{d=1}^p \left( e^{-i (p-d)\Delta t\, H_{B,1}} \; e^{-i d \Delta t\, \bar{H}_C} \right) \ket{+}^{\otimes n}.
\end{equation}
This discretized circuit corresponds to the QAOA ansatz with a fixed parameter schedule determined by the adiabatic path, providing a rigorous theoretical baseline for QAOA on max-$k$-SSAT at quadratic depth.

 % 首先说明变分化的第一个初衷，维持绝热但降低复杂度，我们可以得到第一个基准线，最优绝热路径，维持绝热情况下的参数，这里不知道我应该怎么说
However, adopting a uniform linear schedule does not fully exploit the local structure of the adiabatic evolution, and may therefore require an unnecessarily large number of iterations. To address this issue, we consider a segmented adiabatic evolution scheme, in which the full evolution path is divided into multiple intervals, $\{H(s_0), H(s_1), \ldots, H(s_p)\}$. Within each segment $(H(s_{d-1}), H(s_d)]$, the adiabatic evolution is treated independently.

Unlike the global evolution scheme, where the total runtime is constrained by the overall minimum spectral gap $g_0$, the segmented formulation allows the evolution time in each interval to depend only on the local minimum gap $g(s)$ and the corresponding transition matrix element $\varepsilon_0$ within that segment, according to the adiabatic condition in Eq.~\eqref{eq_adiabatic_condition}. Since many segments typically possess spectral gaps significantly larger than $g_0$, the required runtime for those intervals can be substantially reduced. Consequently, the total evolution time becomes shorter, leading to a corresponding reduction in the circuit depth after discretization.

If the number of segments is sufficiently large, each interval can be made correspondingly small. For a short segment $\left(s_d-\Delta s,s_d\right]$, the minimal evolution time required to maintain adiabaticity scales as $T_d \sim \varepsilon(s_d)g^{-2}(s_d).$ Taking the continuous limit $\Delta s \to 0$, the total evolution time over the entire interval $\left[0,1\right]$ satisfies
\begin{equation}\label{eq_evolution_time}
	T \sim \int_0^1 \varepsilon(s)g^{-2}(s)\,{\rm d}s.
\end{equation}
To adapt the evolution schedule to minimize computational resources, we promote the coefficients of the two Hamiltonians to independent variational parameters, yielding the circuit
\begin{equation}\label{eq_QAOA}
	\ket{\bm{\gamma}, \bm{\beta}} = \prod_{d=1}^p \left( e^{-i \beta_d \, H_{B,1}} \; e^{-i \gamma_d \, \bar{H}_C} \right) \ket{+}^{\otimes n},
\end{equation}
which coincides with the standard QAOA ansatz.

From this perspective, QAOA can be interpreted as a variational generalization of the Trotterized adiabatic $k$-local quantum search, where the fixed schedule $s=t/T$ is replaced by a flexible high-dimensional parameter sequence $(\bm{\gamma},\bm{\beta})$. These parameters are variationally optimized to maximize the objective expectation value
\begin{equation}\label{eq_QAOA_opt}
	E(\bm{\gamma}, \bm{\beta}) = \left\langle \bm{\gamma}, \bm{\beta} \right| H_C \left| \bm{\gamma}, \bm{\beta} \right\rangle,
\end{equation}
which may be viewed as variationally tracking the optimal adiabatic trajectory $(\boldsymbol{\gamma}^*,\boldsymbol{\beta}^*)$.

Although variational optimization allows the parameters to better conform to the adiabatic condition (see Eq.~\eqref{eq_QAOA_opt}) and thereby reduces the required depth, the scaling of the spectral gap $g_k(s) = \Theta(n^{-1})$ implies that maintaining strict adiabaticity still requires a depth of $p = \Omega(n^2)$. In the NISQ regime, practical implementations are typically restricted to much shallower circuits, with depth $p = \mathcal{O}(n)$. In such shallow-depth regimes, the adiabatic condition is inevitably violated. Understanding how to preserve algorithmic performance under these compressed-depth constraints is therefore a central focus of the following discussion.

\section{Smooth adiabatic-manifold parameterization}\label{sec_SAMP}

In the previous section, we established the effectiveness of QAOA on random max-$k$-SSAT by connecting it to the adiabatic $k$-local quantum search framework. However, practical implementations in the NISQ regime are typically restricted to much shallower circuits with depth $p=\mathcal{O}(n)$. This section focuses on such situations to understand how the effectiveness inherited from the adiabatic framework can persist under such substantial depth compression, and how the corresponding variational parameters can be optimized efficiently.

To this end, we first study how the compressed adiabatic evolution induces an effective variational structure in the QAOA parameter space, leading to the emergence of a low-dimensional manifold that preserves essential adiabatic characteristics. Building on this observation, we further analyze the geometric and statistical properties of the resulting adiabatic manifold on random max-$k$-SSAT instances. Motivated by the smoothness and stability of this structure, we then introduce an adiabatic-manifold-based reparameterization scheme together with a hierarchical refinement strategy, aiming to reduce the optimization complexity while maintaining algorithmic performance.

%%%%% 这部分还缺少的内容 %%%%%
% 给出这个流形存在的理论证明，比如设计压缩深度的搜索
%\subsection{Convergence and gradient analysis:} 
% 展示在该流形上的梯度行为，证明避开了贫瘠高原。
%\subsection{Sublinear complexity analysis:}
% 从理论上解释为什么自由度限制导致了亚线性优化成本。

\subsection{From adiabatic baseline to variational manifold}
\label{subsec_variational_manifold}

This section aims to explore the origin of algorithmic effectiveness on random max-$k$-SSAT under the constraint of circuit depth $p = \mathcal{O}(n)$. Recalling from Sec.~\ref{subsec_statistical_analysis} that random max-$k$-SSAT instances converge to the $k$-local search problem, we first examine the algorithm's performance on this idealized problem. Under this depth constraint, the known algorithm corresponds to the $k$-local quantum search (see Eq.~\eqref{eq_universal_mixer_circuit}). While this approach is effective for the ideal $k$-local search, our analysis indicates that it exhibits poor robustness against the perturbations introduced by max-$k$-SSAT, making practical deployment challenging.

Given these limitations, we turn our attention to the adiabatic $k$-local quantum search. As discussed in Sec.~\ref{sec_theory}, the original adiabatic protocol requires a depth of $\Theta(n^2)$. Compressing the circuit depth to $p = \mathcal{O}(n)$ inevitably leads to diabatic transitions to excited states—a phenomenon referred to as \textit{adiabatic leakage}. Understanding how to mitigate the impact of such leakage is crucial for preserving algorithmic performance. As a starting point, we consider the simplest approach: retaining the linear schedule $s = t/T$ while compressing the Trotter discretization to $O(n)$ steps.

When the Trotter discretization of the adiabatic evolution is compressed from $O(n^2)$ to $O(n)$ steps, the effective evolution time per layer must be increased. To accommodate this compressed circuit depth, we introduce a tunable parameter $\rho$ that scales the Trotter time steps. Under this formulation, the circuit is expressed as
\begin{equation}\label{eq_Trotterized_QS}
	\ket{\psi} = \prod_{d=1}^p \left( e^{-2i \frac{n+1-d}{p+1}\, \rho \pi \, H_{B,1}} \; e^{-2i \frac{d}{p+1}\, \rho \pi \, \bar{H}_C} \right) \ket{+}^{\otimes n},
\end{equation}
where the inclusion of $n+1$ ensures that the boundary terms do not render the evolution trivial, thereby preserving as much of the dynamics as possible within the constrained depth.

In this regime, the adiabatic condition is violated, and the characteristics of the $k$-local quantum search become prominent. Specifically, each Trotterized iteration in Eq.~\eqref{eq_Trotterized_QS} corresponds to the system Hamiltonian $\mathcal{H}_k(s)$ interpolated at a particular value of $s$. According to Lemma~\ref{lemma_gap}, the spectral gap of $\mathcal{H}_k(s)$ generally scales as $\Theta(n^{-1})$, implying that the gap of each iteration Hamiltonian is also $\Theta(n^{-1})$. With $O(n)$ iterations, the highest and next-highest energy states remain sufficiently well-separated, and combined with the proofs in Appendix~\ref{apsubsec_proof_search_on_search} on $k$-local quantum search, this ensures the algorithm retains its search capability.

This observation indicates that the algorithm defined in Eq.~\eqref{eq_Trotterized_QS} exhibits a hybrid behavior, blending both search and adiabatic properties. While this allows the algorithm to operate with only $O(n)$ iterations and provides some robustness against the perturbations introduced by max-$k$-SSAT, the precise degree of this robustness remains unclear. To begin addressing the former, we examine its performance on the $k$-local search problem under $O(n)$ depth. Using the circuit in Eq.~\eqref{eq_Trotterized_QS}, we simulate 3-local search instances with $n=16, 18, 20$ and $p=n$. The resulting amplitudes of the target state $\ket{t}$ for various choices of $\rho$ are shown in Tab.~\ref{tab_Trotterized_QS}.

\begin{table}[h]
	\centering\footnotesize
	\caption{\small Target state amplitudes for selected parameters $\rho$ after applying Eq.~\eqref{eq_Trotterized_QS} to 3-local search instances with $n=16,18,20$ and $p=n$.}
	\label{tab_Trotterized_QS}
	\tabcolsep 2.75pt 
	\begin{tabular*}{0.43\textwidth}{ccccccccc}
		\toprule
		$n \backslash \rho$ & 1 & 2 & 3 & 4 & 5 & 6 & 7 & 8 \\
		\bottomrule\addlinespace[0.3em]
		16 & 0.090 & 0.577 & 0.887 & 0.922 & 0.614 & 0.193 & 0.118 & 0.103 \\
		18 & 0.066 & 0.524 & 0.848 & 0.877 & 0.799 & 0.519 & 0.056 & 0.097 \\
		20 & 0.048 & 0.475 & 0.806 & 0.851 & 0.790 & 0.584 & 0.244 & 0.022 \\
		\bottomrule
	\end{tabular*}
\end{table}

The simulation results indicate that as $\rho$ increases, the amplitude of the target state $\ket{t}$ exhibits fluctuations within a certain range, following an overall trend of initially rising and then falling, reaching a pronounced peak. For instance, when $n=20$ and $\rho=7.5$, the amplitude reaches 0.918. Similar behavior is observed when the circuit depth $p$ is varied, with the corresponding optimal value of $\rho$ shifting accordingly. 

These observations suggest the existence of a variational manifold: compressing or stretching the evolution does not eliminate it, but merely scales it along with the Trotter discretization. The key question is whether this manifold retains sufficient adiabatic character to preserve algorithmic effectiveness in the presence of perturbations introduced by the max-$k$-SSAT problem.This issue will be examined in detail in the following subsection.

%These observations suggest that the parameter space is not unstructured, but instead contains a low-dimensional and smoothly connected effective region. In particular, the concentration of high-performance parameters near the diagonal indicates that the algorithm remains sensitive to a balance between the mixer Hamiltonian and the problem Hamiltonian, even in the compressed-depth regime.

\subsection{Emergence of adiabatic manifold}
\label{subsec_adiabatic_manifold}

Building on the variational manifold introduced in the previous subsection, we further investigate its behavior on random max-$k$-SSAT instances. To facilitate this analysis, we introduce additional variational degrees of freedom into the system Hamiltonian while retaining the linear schedule $s = t/T$. The resulting parameterized Hamiltonian is defined as
\begin{equation}\label{eq_paraform_two}
	H(t) = \frac{t}{T} f_\gamma \, \bar{H}_C + \frac{T-t}{T} f_\beta \, H_{B,1},
\end{equation}
where $f_\gamma$ and $f_\beta$ are tunable parameters that independently control the relative strengths of $\bar{H}_C$ and $H_{B,1}$ throughout the evolution.

Using this parameterized Hamiltonian with Trotterization in the form of Eq.~\eqref{eq_Trotterized_QS}, we evaluate the algorithm on random instances of max-3-SSAT generated from the model $F_s(12, m, 3)$, with circuit depth set to $p = n$. The clause density $m$ is chosen to be linear in $n$ and close to the satisfiability threshold, corresponding to the hardest region of $k$-SAT (see Appendix~\ref{apsubsec_average_case}). We explore a range of values for the tunable parameters $(f_\gamma, f_\beta)$ between 1 and 7. For each parameter pair, the circuit is measured over 10,000 shots. The resulting probabilities of obtaining the target state $\ket{t}$ are summarized in the heatmap shown in Figure~\ref{fig_exp4ng}. Additional simulation results are provided in Appendix~\ref{apsubsec_extra_simulation}. 

\begin{figure}[!t]
	\centering
	\footnotesize
	{\includegraphics[width=0.9\linewidth]{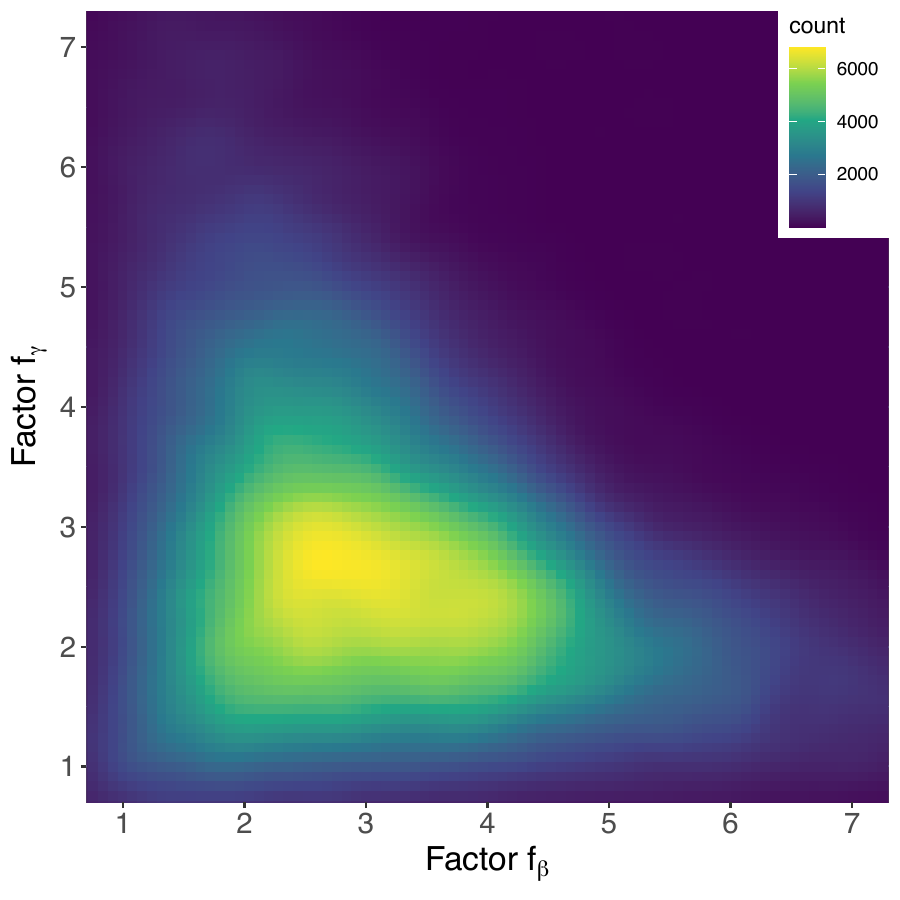}}
	\caption{\small Distribution of measurement counts for the target state $\ket{t}$. For each pair of $(f_\gamma, f_\beta)$, the final state is measured 10,000 times, and the count represents the number of times $\ket{t}$ is observed. The simulation is performed on a random instance from $F_s(12, m, 3)$, with the circuit depth set to $p = n$.} 
	%\hspace{0.2\textwidth}This value is directly proportional to the ground state success probability.
	\label{fig_exp4ng}
\end{figure}

In the simulations, the amplitude of the target state $\ket{t}$ reaches a pronounced peak and exhibits a characteristic behavior of first increasing and then decreasing as functions of both $f_\gamma$ and $f_\beta$. As the circuit depth $p$ varies, the shape and location of this high-performance region may shift, yet the peak structure persists provided that $p$ is not too small. This behavior closely resembles the phenomenon observed in the $k$-local quantum search problem, suggesting the existence of a variational manifold that preserves sufficient adiabatic character to suppress the perturbations introduced by max-$k$-SSAT. We refer to this structure as the \textit{Smooth Adiabatic Manifold}.

Furthermore, the heatmaps reveal a clear structural pattern in the parameter landscape. A pronounced increase of the target-state probability is observed near the diagonal direction $f_\gamma \approx f_\beta$, along which the performance rapidly approaches its optimum. In contrast, deviations away from this direction lead to noticeably flatter gradients and reduced success probability. Moreover, once $(f_\gamma,f_\beta)$ exceed a certain threshold, the high-performance region collapses abruptly into a nearly flat landscape with significantly degraded amplitudes.

This phenomenon is also consistent with the theoretical analysis. In the original $k$-local quantum search framework,  the symmetry between the mixer Hamiltonian and the problem Hamiltonian is evident. Although the universal-mixer formulation for max-$k$-SSAT breaks this symmetry, the spectral concentration property $\bar H_C \rightarrow H_k$ implies that the effective strengths of $\bar H_C$ and $H_{B,1}$ remain of the same order. Consequently, the optimal parameter region continues to align approximately along the symmetric direction $f_\gamma \approx f_\beta$, thereby preserving a smooth adiabatic structure in the variational landscape.

Motivated by this observation, it is natural to follow the form of Eq.~\eqref{eq_Trotterized_QS} by factoring out the overall energy scale (or equivalently the effective evolution time), while introducing a single parameter $\theta$ to control the relative contributions of the two Hamiltonians. Under this formulation, the variational quantum circuit ican be viewed as a Trotterized evolution generated by the effective system Hamiltonian
\begin{equation}\label{eq_H_theta_t}
	\bar{H}(\theta,t) = \frac{t}{T} \sin{\theta} \, \bar{H}_C + \frac{T-t}{T} \cos{\theta} \, \bar{H}_B,
\end{equation}
where the total evolution time is given by $T = 2p\rho\pi$, and the discretization interval is chosen as $\Delta t = 2\rho\pi$. Comparing this construction with a depth-$p$ QAOA circuit, the associated variational parameters $(\gamma_d,\beta_d)$ at layer $d$ ($1 \le d \le p$) are given by
\begin{equation}\label{eq_linear_para}
	\gamma_d = \frac{2 d \pi}{p+1} \, \rho \, \sin{\theta}, \qquad
	\beta_d = \frac{2 (p-d+1) \pi}{p+1} \, \rho \, \cos{\theta}.
\end{equation}

Based on this circuit construction, we further perform numerical experiments on a larger collection of random instances generated from the model $F_s(n,m,3)$. In particular, we investigate the regions in the parameter space $(\theta,\rho)$ that maximize the amplitude of the target state. The resulting distribution of the optimal parameters $(\theta^*,\rho^*)$ over random max-$k$-SSAT instances is illustrated in Figure~\ref{fig_exp4id}.

\begin{figure}[!t]
	\centering
	\footnotesize
	{\includegraphics[width=0.9\linewidth]{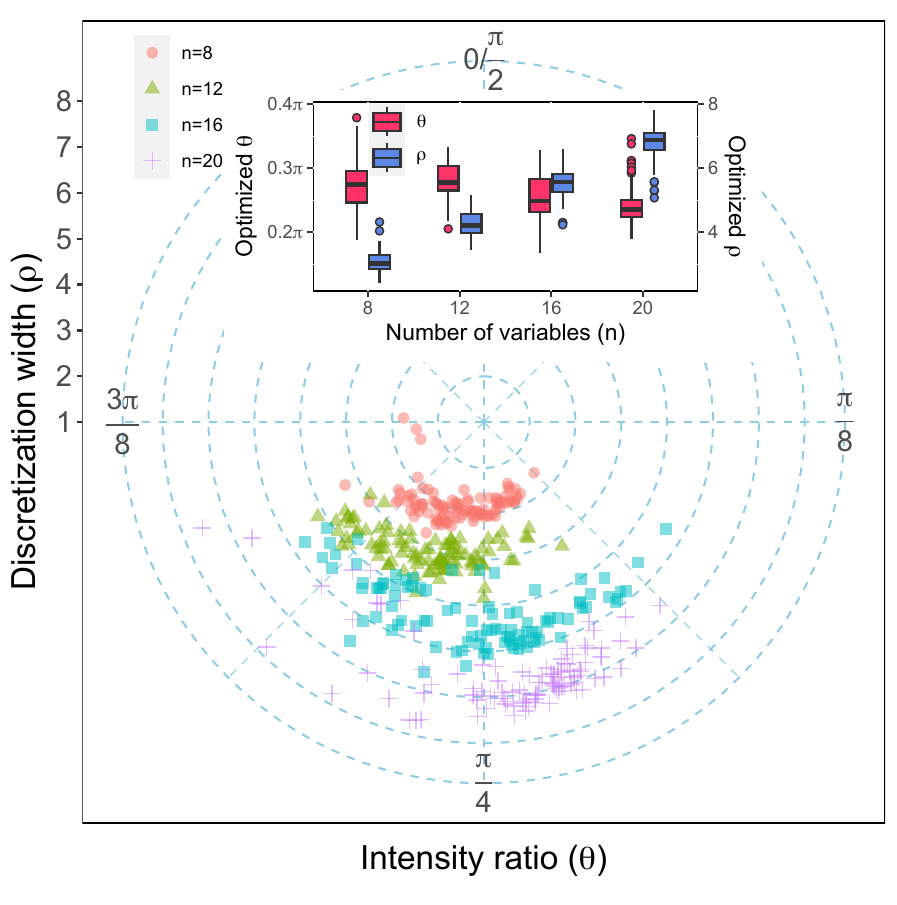}}
	\caption{\small Distribution of optimal $(\theta^*, \rho^*)$ for 100 random instances of $F_s(n, m, 3)$, visualized in polar coordinates. Results for $n = 8, 12, 16, 20$ are indicated by red circles, green triangles, blue squares, and purple crosses, respectively. The inner plot shows that $\theta^*$ values cluster around $\pi/4$, while $\rho^*$ increases approximately linearly with $n$.} %To clearly illustrate the underlying distribution, three idealized conditions are adopted that are not attainable in real applications: (i) the problem Hamiltonian $H_C$ is normalized using the exact $G_0$; (ii) the fidelity is used as the optimization objective; and (iii) the optimal parameters are determined via exhaustive search.
	\label{fig_exp4id}
\end{figure}

In the simulations, although the optimal value $\theta^*$ may deviate slightly from $\pi/4$, it consistently remains concentrated around this region. This observation further supports the existence of an approximately symmetric balance between the mixer Hamiltonian and the problem Hamiltonian within the effective variational manifold. For the parameter $\rho$, larger values correspond to longer effective evolution times within each discretized interval. Since QAOA circuits are typically implemented at relatively shallow depths, an approximately linear increase in $\rho$ effectively compensates for the compression of the adiabatic $k$-local quantum search evolution from the original $\Theta(n^2)$ depth scale down to $O(n)$.

These results provide further evidence that this variational manifold retains sufficient adiabatic character even under the perturbations introduced by max-$k$-SSAT. Meanwhile, the concentration of optimal parameters within a narrow and continuously connected region suggests that the corresponding variational landscape is not random or fragmented, but instead possesses a highly smooth and navigable structure.

The persistence of this manifold arises from a combination of the spectral concentration of the Hamiltonian discussed in Sec.~\ref{subsec_statistical_analysis} and the inherent effectiveness of the adiabatic $k$-local quantum search. In the remainder of this section, we focus on analyzing and exploiting this manifold, exploring strategies to maintain algorithmic performance while efficiently navigating the cost landscape of parameter optimization.

\subsection{Adiabatic-manifold-based reparameterization}\label{subsec_para_space}

The discussions in the previous subsections primarily focused on the existence and structural properties of the low-dimensional adiabatic manifold. However, in the NISQ regime, practical limitations such as further depth compression and hardware noise can significantly reduce the amplitude of the target state after the evolution. As a result, variational parameter optimization remains essential for maintaining algorithmic performance.

As analyzed in Sec.~\ref{subsec_connecting_qaoa}, the optimization of the variational parameters $\boldsymbol{\gamma}$ and $\boldsymbol{\beta}$ allows the discrete evolution to better approximate the optimal adiabatic trajectory $(\boldsymbol{\gamma}^*,\boldsymbol{\beta}^*)$. Moreover, when the circuit depth becomes significantly insufficient to satisfy the adiabatic condition, variational optimization can additionally act as a ``healing'' mechanism that partially compensates for the adiabatic leakage induced by finite-depth discretization, noise, and perturbations.

Nevertheless, directly optimizing all QAOA parameters in the original high-dimensional parameter space remains computationally expensive and often suffers from poor trainability. Motivated by the smoothness and concentration properties of the adiabatic manifold established in the previous subsections, we instead seek a reparameterization adapted to the manifold structure itself. The central idea is to transform the variational parameter space into a representation that is more naturally aligned with the effective low-dimensional geometry of the adiabatic manifold, thereby enabling more efficient optimization while preserving algorithmic performance.

To formalize this reparameterization, we interpret the discrete QAOA parameters as samples drawn from an underlying continuous adiabatic trajectory. According to the analysis of the optimal adiabatic trajectory in Sec.~\ref{subsec_connecting_qaoa}, the system Hamiltonian can be generalized as
\begin{equation}
	H(t) = f_C(t)\,\bar{H}_C + f_B(t)\,H_{B,1}.
\end{equation}
As the evolution proceeds, the system Hamiltonian evolves smoothly from the mixer Hamiltonian toward the problem Hamiltonian. The pair of functions $\bm{f}(t)=(f_C(t),f_B(t))$ therefore completely characterizes the adiabatic trajectory.

Importantly, the optimal adiabatic trajectory $\bm{f}^*(t)$ is expected to remain continuous with respect to time. Intuitively, if an ``optimal'' schedule contained discontinuities, one could always construct a smoother alternative through local interpolation, thereby reducing the magnitude of $\frac{{\rm d}H(t)}{{\rm d}t}$, which in turn relaxes the adiabatic constraint and decreases the required evolution time. Consequently, the effective variational manifold induced by the compressed adiabatic evolution is naturally expected to inherit a smooth geometric structure.

The remaining question is therefore how to parameterize such adiabatic trajectories in a manner that better aligns with the geometry of the smooth adiabatic manifold. Motivated by the previous analysis of the linear schedule described in Eq.~\eqref{eq_H_theta_t} and Eq.~\eqref{eq_linear_para}, where the manifold exhibits favorable structural properties, we replace the original parameter space $(\bm{\gamma},\bm{\beta})$ with an alternative representation $(\bm{\tau},\bm{\theta})$ that is more naturally adapted to the manifold geometry. After factoring out the global scaling component, the correspondence between the two parameterizations can be written as
\begin{equation}\label{eq_repara}
	\gamma_d=\frac{2\pi d}{p+1}\tau_d\sin{\theta_d}, \qquad
	\beta_d=\frac{2(p+1-d)\pi}{p+1}\tau_d\cos{\theta_d},
\end{equation}

Here, we present this reparameterization primarily from an intuitive and phenomenological perspective motivated by the observed manifold structure. A more detailed derivation and theoretical analysis of this parameterization framework are deferred to Appendix~\ref{apsec_analysis_am_repara}. 

Importantly, the two parameters admit clear physical interpretations. The angle parameter $\theta_d$ determines the relative strength between the normalized Hamiltonians $\bar{H}_C$ and $\bar{H}_B$, thereby characterizing the local evolution direction along the adiabatic manifold. Meanwhile, the scaling parameter $\tau_d$ controls both the effective discretization interval in the Trotterized adiabatic evolution and the total evolution time associated with each QAOA layer.

\subsection{Hierarchical refinement strategy}\label{subsec_para_AP}

Building on the manifold-based parameter space $(\bm{\tau},\bm{\theta})$ introduced in the previous subsection, we now investigate how the smooth adiabatic manifold can be further exploited for parameter initialization and optimization. The primary objective of this subsection is to reduce the complexity of variational parameter optimization while simultaneously improving the effectiveness and stability of the resulting QAOA evolution.

\subsubsection{Parameter initialization}

We first discuss the initialization strategy. Based on the observations in Sec.~\ref{subsec_adiabatic_manifold}, the linear-schedule regime exhibits favorable performance and therefore provides a natural starting point for parameter initialization. In particular, according to Fig.~\ref{fig_exp4ng}, one direct approach is to initialize the parameters within the high-performance region. However, the precise location of this optimal region generally shifts with the system size $n$, and therefore typically requires prior information, for example obtained from corresponding $k$-local quantum search or random instances.

Another more general consideration is to initialize in regions with sufficiently large gradients, which are more favorable for subsequent variational optimization. Since both the mixer Hamiltonian and the problem Hamiltonian are normalized, their effective strengths are expected to remain of comparable magnitude during the evolution. This directly motivates the symmetric choice $\theta=\frac{\pi}{4}$, which balances the contributions of the two Hamiltonians along the adiabatic trajectory.

The choice of $\rho$ can likewise be understood from the role of the problem Hamiltonian during the evolution. Specifically, the problem information enters the circuit through the unitary $e^{-i\gamma \bar{H}_C}$, which applies phase factors $e^{-i\gamma \bar{f}(x)}$ to the computational basis states $\ket{x}$, where $\bar{f}(x)$ denotes the normalized objective function. The target state is distinguished through its phase evolution relative to the remaining basis states. Consequently, maximizing distinguishability requires the target phase $e^{-i\gamma \bar{f}(x)}$ to remain sufficiently separated from the phases associated with non-target states. Due to the intrinsic $2\pi$ periodicity of quantum phases, it is therefore natural to constrain $\gamma \bar{f}(x)$ approximately within the interval $[0,2\pi)$. Under the adiabatic-manifold parameterization adopted here, this consideration naturally leads to the choice $\rho=\sqrt{2}$. Substituting these initialization parameters into Eq.~\eqref{eq_linear_para} then yields the corresponding initialization in the original parameter space $(\bm{\gamma},\bm{\beta})$.

\subsubsection{Parameter tuning}

By employing a carefully designed initialization, the resulting parameters can be conceptualized as a rough approximation of the adiabatic manifold under linear schedule. The goal of the parameter setting process is to reduce the deviation between the estimated path $\bm{f}(t)$ and the target $\bm{f}^*(t)$ by minimizing the discrepancy $\Delta \bm{f}(t)= \bm{f}(t) - \bm{f}^*(t)$. Specifically, we aim to minimize the norm $\left\| \Delta \bm{f}(t)\right\|$, where the norm is defined as the maximum absolute deviation, i.e., $\left\| g(x) \right\|= \left\| \max_x \{g(x)\}\right\|$, due to the continuity of the function.

As discussed in Sec.~\ref{subsec_adiabatic_manifold}, a substantial gradient is typically observed near the initialization, which effectively guides the parameter optimization process. Therefore, aside from the desired precision threshold $\epsilon$, the optimization cost is primarily determined by the initial deviation  $\left\| \Delta \bm{f}(t)\right\|$. If the initial deviation is bounded, the optimization is expected to converge within a constant number of gradient steps, with each gradient evaluation requiring $\mathcal{O}(p)$ circuit executions.

According to the distribution of the optimal parameters $(\theta^*, \rho^*)$ illustrated in Figure~\ref{fig_exp4id}, the deviation $\Delta \bm{f}(t)$ is generally bounded when the Hamiltonian is properly normalized. Furthermore, if the overall parameter magnitudes are scaled to be $\mathcal{O}(n)$, the optimization can proceed with constant gradient steps, regardless of problem size. As a result, in the idealized setting, the overall optimization cost becomes approximately linear in $p$, providing a baseline for further improvement.

Building on this, we consider a more structured parameter adjustment strategy inspired by the idea of progressive refinement. Since the parameterized adiabatic passage aims to use $p$ sampling points to approximate the continuous optimal path $\bm{f}^*(t)$, it is more practical to begin with a small number of points that capture the coarse structure of the path and then incrementally refine the approximation by increasing the resolution. This idea is conceptually similar to the FOURIER strategy proposed in Ref.~\cite{Zhou2020}.

Specifically, FOURIER introduces an indicator $q$ to represent the current number of sampling points and decomposes $\bm{f}(t)$ into $q$ components using a cos/sin Fourier transform. As $q$ increases linearly, the optimal parameters from the previous stage are reused as the initialization for the next. On average, this results in a per-iteration deviation $\left\| \Delta \bm{f}_{d}(t)\right\|$ for the $d$-th component of roughly $\left\| \Delta \bm{f}(t)\right\|/q$, thereby reducing the optimization cost at each step. However, since $q$ grows linearly, the total number of stages is $\mathcal{O}(n)$, resulting in an overall cost that remains superlinear.

To address this, we refine the strategy by doubling the number of sampling points at each iteration instead of increasing them linearly. In the $l$-th iteration, $2^l$ sampling points are used, and the initial parameters are obtained by interpolating those from the previous stage with $2^{l-1}$ points. Assuming sufficient continuity in the parameter landscape, the deviation $\left\| \Delta f_j(t)\right\|$ for the $j$-th interval (with $1 \le j \le 2^l$) can be approximately reduced to $\left\| \Delta f(t)\right\|/2^l$. As a result, only logarithmic iterations are required, potentially reducing the overall optimization cost to $\mathcal{O}(\log(p))$. The corresponding progressive refinement strategy is detailed in Algorithm~\ref{alg2}.

% explain about each step
In this algorithm, the interpolation routine $Interp(\bm{v}, L)$ is flexible; in this paper, we use cubic spline interpolation. In Line~1--2, the parameters $\theta_0 = \pi/4$ and $\tau_0 = \sqrt{2}$ are initialized and then optimized using a QAA-inspired setting to obtain quasi-optimal values $(\theta_0^*, \tau_0^*)$. In Line~3, the Hamiltonians $\bar{H}_C$ and $\bar{H}_B$ are rescaled according to these optimized parameters to ensure consistent parameter magnitudes across different resolutions. In Line~4, the number of interpolation stages is set as $T_u = \lfloor \log_2 p \rfloor$, and the parameter vectors are initialized as $\bm{\theta} \leftarrow (\pi/4)$ and $\bm{\tau} \leftarrow (1)$. In the while loop (Line~5 to Line~10), the number of sampling points is doubled by interpolating $\bm{\theta}$ and $\bm{\tau}$. These interpolated vectors are optimized by maximizing the expected value $\langle \bm{\theta}', \bm{\tau}' \mid H_C \mid \bm{\theta}', \bm{\tau}' \rangle$, where $\bm{\theta}', \bm{\tau}'$ are upsampled to full length $p$. The resulting parameters are used to initialize the next round, and the process continues until the final resolution $p$ is reached.

\begin{algorithm}[H]
	\caption{Hierarchical refinement strategy parameter tuning}
	\label{alg2}
	\small
	\begin{algorithmic}[0]
		\Require ~~\\
		Hamiltonian $\bar{H}_C$, $H_{B,1}$; 
		interpolation routine $Interp(\bm{v}, L)$ that interpolates vector $\bm{v}$ to length $L$
		\Ensure ~~\\
		Quasi-optimal parameters $(\bm{\theta}, \bm{\tau})$ for depth-$p$ QAOA
	\end{algorithmic}
	\begin{algorithmic}[1]
		\State Initialize $\theta_0 = \pi/4$, $\tau_0 = \sqrt{2}$
		\State Optimize $(\theta_0, \tau_0)$ using QAA-inspired setting to obtain $(\theta_0^*, \tau_0^*)$
		\State $\bar{H}_C \leftarrow \sqrt{2} \tau^*_0 \sin\theta^*_0 \cdot \bar{H}_C$, $H_{B,1} \leftarrow \sqrt{2} \tau^*_0 \cos\theta^*_0 \cdot H_{B,1}$
		\State $T_u \leftarrow \left\lfloor \log_2 p \right\rfloor$, $\bm{\theta} \leftarrow (\pi/4)$, $\bm{\tau} \leftarrow (1)$
		\While{$T_u > 0$}
		\State $T_u \leftarrow T_u - 1$
		\State $\bm{\theta} \leftarrow Interp(\bm{\theta}, \left\lceil p / 2^{T_u} \right\rceil)$, $\bm{\tau} \leftarrow Interp(\bm{\tau}, \left\lceil p / 2^{T_u} \right\rceil)$
		\State Construct $\bm{\theta}' = Interp(\bm{\theta}, p)$ and $\bm{\tau}' = Interp(\bm{\tau}, p)$
		\State Optimize $\bm{\theta}, \bm{\tau}$ to maximize $\left\langle \bm{\theta}', \bm{\tau}' \right| H_C \left| \bm{\theta}', \bm{\tau}' \right\rangle$
		\EndWhile
	\end{algorithmic}
\end{algorithm}

\section{Numerical results on 3-SAT}\label{sec_results}

This section presents a comparative analysis of the parameter optimization framework proposed in this work and several existing methods, including the TQA-based initialization from~\cite{Sack2021} and the FOURIER heuristics introduced in~\cite{Zhou2020}, with particular emphasis on the observed logarithmic scaling of the optimization cost.

To ensure fair and meaningful comparison, we adopt a standardized experimental setting across all methods. In particular, we fix the QAOA depth at $p = n$, where $n$ is the number of variables, and evaluate algorithmic performance using success probability and optimization cost. Optimization is performed using the BFGS algorithm, and random 3-SAT instances are generated near the satisfiability threshold. Details of the simulation setup, including circuit depth justification, instance generation protocol, evaluation metrics, and cost accounting, are provided in Appendix~\ref{app_exp}.

\subsection{Performance benchmarking}
\label{subsec_benchmarking}

% show the similarity of SAMP parameter setting with FOURIER heuristic strategy
The SAMP method optimizes $2p$ parameters and shares similarities with the FOURIER heuristic in its attempt to reduce optimization cost by exploiting parameter continuity. Specifically, the FOURIER method restricts the degrees of freedom in the parameter space to a value $q$, from which the full set of parameters is reconstructed via a cos/sin Fourier transform. Analogously, the SAMP method introduces an indicator $T_u$ to control the degrees of freedom and generates the full parameter set through interpolation. 

While both methods may produce parameter sets with limited expressivity if terminated early, for the purpose of fair performance comparison, we consider the fully optimized 2$p$-dimensional parameter space in both cases. The detailed steps of the FOURIER heuristic under this setting are provided in Appendix~\ref{app_FOUR}.

% show the difference and present the simulation
Notably, the FOURIER heuristic is primarily designed for settings without reliable initial parameters, often requiring $\mathcal{O}(n)$ iterations in its outer loop. In contrast, the SAMP method inherits the advantage of adiabatic manifold by initializing directly at full depth using a theoretically motivated setting. Moreover, it effectively transforms the parameter space $({\bm\gamma}, {\bm\beta})$ into a more continuous representation, which benefits the subsequent optimization. The TQA method also provides a high-quality initialization through pre-computation and can be directly used to optimize the full parameter set without any additional structure or strategy.

We compare the performance of these three approaches on 100 random instances drawn from $F_s(n,m_n^*,3)$. The results are shown in Figs.~\ref{fig_exp_com}. Figs.~\ref{fig_exp_com1a} and \ref{fig_exp_com1b} show the performance of the three methods on the same set of 100 random instances, evaluated in terms of success probability and optimization cost, respectively.

\begin{figure}[t]
	\centering
	\begin{subfigure}[t]{0.95\linewidth}
		\centering
		\includegraphics[width=\linewidth]{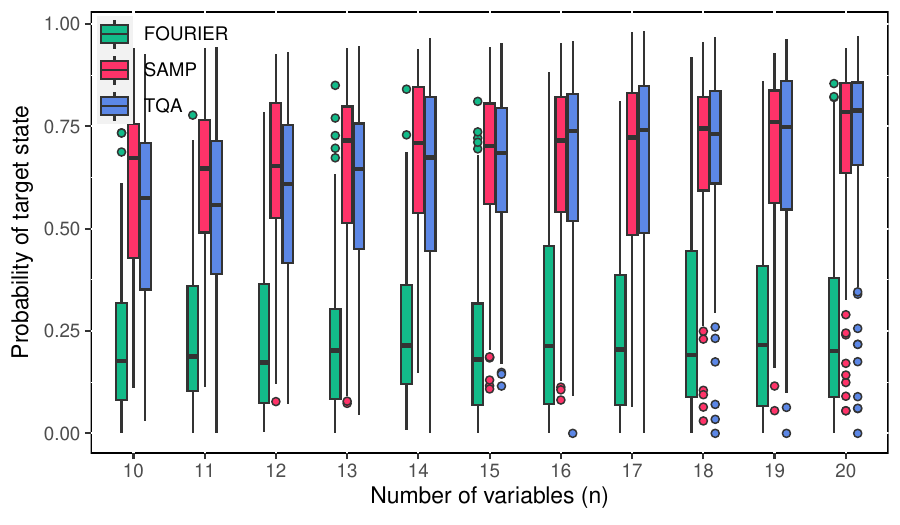}
		\caption{\small Distribution of success probabilities.}
		\label{fig_exp_com1a}
	\end{subfigure}
	\vspace{1ex}
	\begin{subfigure}[t]{0.95\linewidth}
		\centering
		\includegraphics[width=\linewidth]{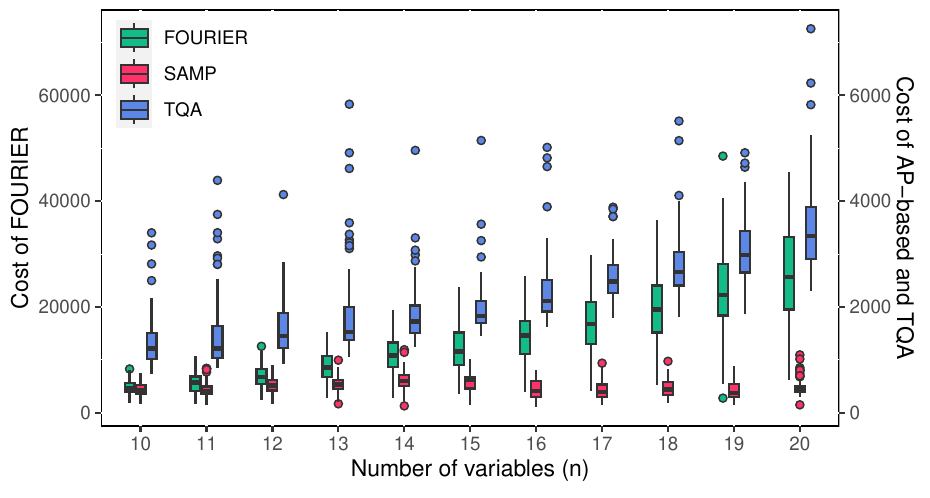}
		\caption{\small Distribution of optimization costs (dual y-axis).}
		\label{fig_exp_com1b}
	\end{subfigure}
	\caption{\small Performance comparison of the FOURIER heuristic (green), the SAMP method (red), and the TQA method (blue) over the same 100 random instances from $F_s(n, m_n^*, 3)$ with $10 \le n \le 20$.}
	\label{fig_exp_com}
\end{figure}

Distribution of optimization costs corresponding to the same 100 instances used in Fig.~\ref{fig_exp_com1a}. The results for the FOURIER heuristic, the SAMP method, and the TQA method are shown in green (left), red (middle), and blue (right), respectively. A dual y-axis is used to better reflect the differences in magnitude between methods.

% analysis of the res of probability
Compared to the FOURIER heuristic, both the SAMP and TQA methods demonstrate superior performance in terms of success probability. The strong results reported for FOURIER on Max-Cut problems in Ref.~\cite{Zhou2020} highlight the critical role of Hamiltonian normalization and parameter initialization---especially for problems like 3-SAT, where the magnitude of the problem Hamiltonian increases rapidly. In the current setting, parameters are initialized only once, which leads to poor performance for heuristic strategies such as FOURIER that lack access to a known feasible starting point. This often results in the optimization process being trapped in local optima.

% analysis of the res of optimization cost
In terms of optimization cost, the FOURIER heuristic exhibits growth slightly exceeding $\mathcal{O}(n^2)$, with average costs of 4853.65 at $n=10$ and 26121.76 at $n=20$. This high cost arises from two factors: the linear increase in circuit depth $q$ and the gradual expansion of the parameter space from 1 to $q$ degrees of freedom. In contrast, the TQA method benefits from a reliable initialization obtained through pre-computation, allowing the elimination of the outer optimization loop and yielding costs slightly above $\mathcal{O}(n)$ (with averages of 1326.99 at $n=10$ and 3523.13 at $n=20$). Notably, the SAMP method achieves even lower costs, approaching sublinear scaling. A more detailed discussion on optimization cost is provided in Sec.~\ref{subsec_perform}.

%This subsection displays more simulation results on general instances 
\subsection{Scaling of optimization cost}\label{subsec_perform}

The average optimization costs of the SAMPmethod for $4 \le n \le 20$ are summarized in Tab.~\ref{tab_count}. As the number of variables $n$ increases, both the circuit depth $p$ and the number of parameters $2p$ naturally grow. Correspondingly, the optimization cost also increases, following an overall sublinear trend—rising from 300.5 at $n=4$ to 546.3 at $n=20$. However, it is worth noting that this growth is not strictly monotonic in the simulations. Instead, the cost tends to increase within certain intervals and exhibits fluctuations when transitioning between them.

\begin{table}[!t] % zhe last three data is not replaced 
	\footnotesize
	\caption{\small The average optimization cost of the SAMP method over 1000 random instances in $F_s(n, m_n^*, 3)$ for $4 \le n \le 20$. The range of $n$ is partitioned into three intervals---$4 \le n \le 7$, $8 \le n \le 15$, and $16 \le n \le 20$---based on the logarithmic indicator $\left\lfloor \log_2{p} \right\rfloor$. }
	\label{tab_count}
	\tabcolsep 2.75pt %space between two columns. ÓÃÓÚµ÷ÕûÁÐ¼ä¾à
	\begin{tabular*}{0.44\textwidth}{ccccccccc}
		\toprule
		$n$ & [4 & 5 & 6 & 7] & [8& 9 & 10 & 11 \\\hline
		cost & $300.5$& $342.9$& $379.5$& $410.6$& $430.4$& $440.3$& $445.0$ & $446.4$\\
		%Harder & Bbb & Ccc & Ddd & 10 & 11 & 12 & 16 & 17 & 18 & 19 & 20\\
		\bottomrule
		%\toprule
		12 & 13 & 14 & 15] & [16 & 17 & 18 & 19 & 20]\\\hline
		$532.2$& $544.3$& $586.0$& $612.2$& $524.9$& $535.7$& $573.5$& $531.7$& $546.3$\\
		%Harder & Bbb & Ccc & Ddd & 10 & 11 & 12 & 16 & 17 & 18 & 19 & 20\\
		\bottomrule
	\end{tabular*}
\end{table}

The performance of the SAMP method is closely related to the outer loop governed by the logarithmic indicator $0 < T_u \le l$, where $l = \left\lfloor \log_2{p} \right\rfloor$. This indicator determines when to increase the number of parameters. Accordingly, the range of variable size $n$ is divided into three intervals: $4 \le n \le 7$, $8 \le n \le 15$, and $16 \le n \le 20$, within which $l$ remains constant while the number of parameters grows gradually. When $n = 2^l$, the interpolation strategy is best aligned with the structure of the algorithm, ensuring a precise doubling of the parameter space in each iteration. As $n$ increases beyond these ideal points, interpolation becomes less effective due to reduced continuity between parameters, leading to increased optimization costs.

Empirical results support this trend: when $n = 4$, $8$, and $16$, the corresponding average optimization costs are 300.5, 430.4, and 524.9, respectively---demonstrating an approximately linear relationship with $l$. This observation suggests that, under this interpolation framework, the optimization cost is more closely correlated with $\log_2{p}$ than with $n$ or $p$ directly.

To further illustrate this dependence, consider the case $n = 16$. As the indicator $T_u$ decreases---effectively doubling the degrees of freedom in each iteration---the cumulative optimization costs are 40.7, 123.1, 294.7, 448.6, and 524.9. These results reveal a near-linear relationship with $l - T_u$, and notably, the final three values closely match the total costs observed for $n = 4$, $8$, and $16$. This consistency indicates that the optimization cost is predominantly determined by the degrees of freedom in the parameter space, rather than directly by $n$ or $p$. Therefore, in practical applications, setting the degrees of freedom to $2^l$ or $2^{l+1}$ can yield near-optimal efficiency, even when $p \ne 2^l$. In such cases, the performance in terms of success probability remains largely unaffected due to the inherent continuity of the optimal adiabatic passage.

% analysis about the continuity of optimized parameters
This efficiency stems from the observed continuity in the optimized parameters, as illustrated in Fig.~\ref{fig_exp6para}. The parameters exhibit smooth profiles reminiscent of trigonometric functions, with $\theta_d$ and $\rho_d$ centered around $\pi/4$ and 1, respectively. This smoothness greatly reduces the variation in $g(s)$ during simulation, resulting in $\max_d{\tau_d}$ being no more than twice $\min_d{\tau_d}$. In contrast, for $g(s)$ itself, the maximum typically differs from the minimum by a much larger factor. These findings are consistent with the theoretical analysis in Sec~\ref{subsec_para_AP}, highlighting the importance of Hamiltonian normalization and parameter continuity. They also suggest that the Sin/Cos Fourier transform is particularly well-suited to exploiting the continuous structure of the parameter space $({\theta}, {\tau})$. While this work adopts a simple interpolation routine for convenience, the underlying continuity points to promising directions for more sophisticated heuristic parameterizations.

\begin{figure}[!t]
	\centering
	\footnotesize
	{\includegraphics[width=1\linewidth]{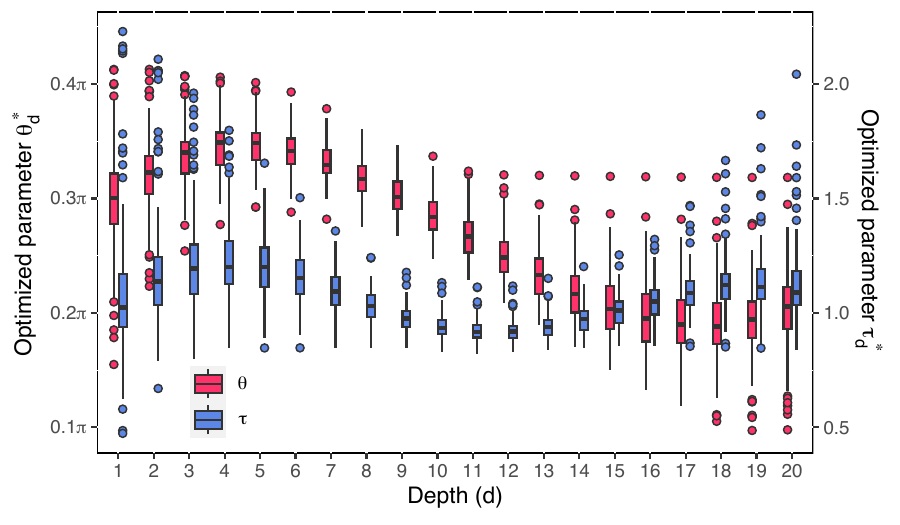}}
	\caption{\small The distribution of the optimized parameters ${\bm \theta}^*$ and ${\bm \tau}^*$ for 100 random instances in $F_s(20,m_n^*,3)$, presented in red on the left and blue on the right, respectively. The optimized parameters generally present superior continuity reminiscent of trigonometric functions, with even the outliers persisting very close to the adjacent parameters. }
	\label{fig_exp6para}
\end{figure}

% Finally give the results checking the running time of circuit. 
In practical implementations, the expectation value is estimated by executing the quantum circuit multiple times with the given parameters, typically requiring around 1000 shots. The complexity of a single circuit run scales as ${\mathcal O}(mp)$, depending on the cost Hamiltonian $H_C$ and the depth $p$. Although this may suggest a high total cost, the actual overhead is mitigated by the preprocessing and parameter initialization steps.

% 之后再改
%As shown in Fig.~\ref{fig_exp_com2a}, QAOA can often yield favorable performance with only a few circuit evaluations in the early stages. Nevertheless, effective parameter setting remains crucial, particularly due to the problem reduction when handling unsatisfiable instances. Given a time constraint of $T = {\mathcal O}({\rm Poly}(n))$ to terminate the optimization process, a well-designed parameter setting method can accelerate convergence to a quasi-optimal solution and reduce the overall error rate. This advantage extends the applicability of the algorithm to general instances of 3-SAT.

% Even with the presence of noisy in simulation, the average required runs remain insufficient for a single expectation calculation, as shown in Table \ref{tab_noisy_run}.

%\subsection{Noise Resilience}

\section{Discussion and generalization}\label{sec_conclusion}
% 这里面记得讨论k-local量子搜索的一般形式。

\subsection{On the Applicability Beyond 3-SAT}\label{sec_app_beyond}% in nearly 1 page

% introduce the framework of method on general COP
While the core analysis in this work is based on max-3-SSAT, the proposed parameter initialization framework is not inherently limited to this setting. In particular, it can be naturally extended to general max-$k$-SAT problems and more broadly to general combinatorial optimization problems. The key idea is to construct a suitable random instance model for the problem of interest, derive the corresponding random problem Hamiltonian $H_C$, and estimate its maximum energy difference. Depending on the problem structure, this may involve techniques such as problem reduction or model approximation. Once the statistical properties of $H_C$ are understood---analytically or empirically---the same parameter setting method introduced for Max-3-SSAT can be applied. We outline such generalizations below through several representative classes of problems.

% introdu about the graph-related problem 
For combinatorial optimization problems on graphs, such as max-cut and max-clique, natural random models like the Erdős–Rényi model can be employed. The eigenvalue distribution of the corresponding problem Hamiltonian can then be derived accordingly. For instance, in Ref.~\cite{Zhou2020}, ,max-cut is studied on 3-regular graphs. Notably, in the case of edge density around 0.5, the maximum eigenvalue of $H_C$ is bounded by the total number of edges, approximately $1.5n$, which remains on the same order as that of $H_B$. As a result, both the FOURIER heuristics exhibit a high probability of reaching the target state.

However, for $k$-regular graphs with larger $k$, or for random graphs such as $G(n, \rho)$, the normalization of the Hamiltonian becomes essential to ensure effective parameter initialization. Moreover, for constrained combinatorial problems like max-clique, the form of the problem Hamiltonian depends on the specific design of the cost function $C(x)$. Nevertheless, its eigenvalue distribution can still be analyzed in a similar manner.

% for COP without known random model
While some combinatorial optimization problems lack widely accepted random models, it is often still possible to construct problem-specific ones. For example, in the set cover problem, a random model can be defined by restricting each subset to contain exactly $k$ elements. Given a universal set $U$ with $n$ elements, each subset $S_j \subset U$ is chosen such that $\left| S_j \right | = k$. Then, by uniformly, independently, and with replacement selecting $m$ subsets from the $C_n^k$ possible options, a random instance model similar to $F(n, m, k)$ can be formulated. Moreover, even when an analytical model is hard to construct or evaluate, empirical approaches remain viable for obtaining approximate results.

To demonstrate that the proposed framework extends beyond random 3-SAT, we further consider combinatorial optimization problems on graphs as a representative class of applications. Although the framework could potentially be applied to a variety of problems, including Max-Clique and even non-graph problems such as Set Cover, we focus on Max-Cut as a concrete example. The graph instances are generated using the Erd\H{o}s--R\'enyi random graph model, with edge density fixed at $0.5$. We then apply the SAMP framework to these random Max-Cut instances, and present the corresponding simulation results in Appendix~\ref{app_cut}.

The numerical results show that the performance of the proposed framework remains largely stable across different problem sizes, while the associated optimization cost is comparable to that observed in the random 3-SAT setting. These observations suggest that the proposed parameterization and optimization framework is not restricted to random max-$k$-SSAT, but can also be effectively extended to Max-Cut problems. More broadly, the stability of the performance across varying system scales indicates the potential applicability of the framework to a wider class of combinatorial optimization problems, provided that suitable random-instance models can be constructed or reasonably approximated.

\subsection{Discussion and conclusion}\label{sec_conclu} % in half page、
% Concludes the method

In this work, we established a rigorous connection between the theoretical foundations of adiabatic quantum evolution and the variational effectiveness of QAOA for $k$-local combinatorial optimization problems. Using the universal-mixer $k$-local quantum search as a canonical reference framework, we showed that the effectiveness of QAOA is not merely the result of heuristic variational optimization, but instead originates from an underlying low-dimensional geometric structure, which we refer to as the \textit{smooth adiabatic manifold}.

Building on this observation, we proposed the SAMP framework, which exploits the geometric regularity inherited from adiabatic evolution to address the long-standing challenge of parameter optimization in the NISQ regime. The proposed framework is constructed upon two complementary ingredients. The first is a statistical normalization of the Hamiltonian derived from the concentration properties of random $k$-SAT models, which enables stable and instance-independent parameter scaling. The second is a geometric reparameterization of the variational parameter space motivated by the continuity and smoothness of optimal adiabatic trajectories. Together, these ingredients provide an effective shortcut to adiabaticity, allowing the algorithm to retain high success probabilities at circuit depths scaling as $p=\mathcal{O}(n)$, despite the $\Theta(n^2)$ depth typically required by standard linear adiabatic evolution.

Our numerical simulations on random 3-SAT instances demonstrate that SAMP substantially reduces the optimization overhead while maintaining strong algorithmic performance. In particular, the optimization cost exhibits a significantly improved scaling behavior with respect to the circuit depth $p$, potentially approaching logarithmic scaling in practical regimes. Beyond computational efficiency, the SAMP framework also provides several conceptual advantages.

First, the framework naturally provides a high-quality initialization strategy for arbitrary circuit depths, avoiding the need for incremental layer-by-layer optimization procedures and enabling direct utilization of the available computational depth. Second, the resulting parameter manifold exhibits strong structural stability across varying system sizes $n$ and clause densities $m$, suggesting the existence of universal geometric features in the variational landscapes of random $k$-local optimization problems. Third, by interpreting QAOA as a variational correction mechanism for adiabatic leakage, the framework offers a physically transparent picture of the underlying quantum dynamics, moving beyond the conventional black-box interpretation of variational optimization.

The emergence of the Smooth Adiabatic Manifold also opens several promising directions for future research. One immediate direction is the development of more expressive manifold parameterizations, including higher-order schedules or machine-learning-assisted manifold tracking techniques, in order to further improve optimization efficiency. In addition, although this work primarily focused on the average-case behavior of random 3-SAT instances, the structural insights developed here may also provide useful tools for studying phase transitions and hardness structures in more challenging optimization problems and general $k$-local Hamiltonian systems. More broadly, the principles underlying SAMP are not restricted to SAT-type problems. Instead, they suggest a general physics-informed framework for constructing variational ans\"atze in quantum optimization and quantum simulation, thereby helping bridge the gap between rigid adiabatic evolution and flexible variational quantum control.

\section*{Acknowledgements}
We acknowledge support by Innovation Program for Quantum Science and Technology (Grant Nos. 2021ZD0302900) and Natural Science Foundation of Jiangsu Province, China (Grant No.BK20220804).

\appendix
% Reset figure and table counters
\renewcommand{\thefigure}{A.\arabic{figure}}
\renewcommand{\thetable}{A.\arabic{table}}
\setcounter{figure}{0}
\setcounter{table}{0}

\section{Preliminaries}\label{sec_preli}

\subsection{Average-case complexity theory and random \textit{k}-SAT problem}\label{apsubsec_average_case}

In Levin's framework of average-case complexity theory~\cite{Levin1986}, a \textit{random problem} is represented by a pair $(\mu,R)$, where $R\subset \mathbb{N}\times\mathbb{N}$ denotes the binary relation over instance-witness pairs $(x,y)$, and $\mu:\mathbb{N}\to[0,1]$ is the cumulative distribution function over instances. The associated probability mass function is given by $\mu'(x)=\mu(x)-\mu(x-1)$, which specifies the probability of generating a particular instance $x$. A random problem is said to be \textit{polynomial on average} if the decision problem
\[
\bar{R}(x)\Leftrightarrow \exists y\,R(x,y)
\]
can be solved in polynomial average time. More precisely, although the runtime $t(x)$ may be exponential in $\left|x\right|$ for certain rare instances, the average normalized runtime is required to remain finite:
\begin{equation}
	\sum_x \mu'(x)t(x)/\left|x\right| < \infty.
\end{equation}
Thus, average-case polynomial complexity permits extremely hard instances, provided that their occurrence probabilities are sufficiently small.

Consider two problems $P_1$ and $P_2$ with corresponding instance-witness relations $R_1$ and $R_2$. A polynomial-time algorithm $f(x)$ is said to reduce $P_1$ to $P_2$ if
\[
\bar{R}_1(x)\Leftrightarrow \bar{R}_2(f(x)).
\]
For random problems, the associated probability distributions must also be taken into account. In Levin's theory~\cite{Levin1986}, the relation $\mu_1\lesssim\mu_2$ is defined if there exists a constant $k$ such that
\[
\forall x,\qquad
\frac{\mu_1'(x)}{\mu_2'(x)}<|x|^k.
\]
A polynomial-time algorithm $f(x)$ reduces a random problem $(\mu_1,R_1)$ to $(\mathfrak{f}[\mu_2],R_2)$ if both $\mu_1\lesssim\mu_2$ and
\[
\bar{R}_1(x)\Leftrightarrow \bar{R}_2(f(x))
\]
hold. Here, $\mathfrak{f}[\mu_2]$ denotes the output distribution induced by $f(x)$,
\[
\mathfrak{f}[\mu_2](x')
=
\sum_{f(x)\le x'}\mu_2'(f(x)).
\]
Under this framework, a random NP problem is called complete if every random NP problem can be reduced to it.

In this work, we focus on random $k$-SAT and its related optimization variants. The standard random $k$-SAT model is denoted by $F(n,m,k)$, which is widely used in probabilistic and average-case analyses. In this model, an instance is generated by sampling $m$ clauses independently and uniformly with replacement from the set of all $2^kC_n^k$ possible $k$-clauses~\cite{Achlioptas2006}. Under this distribution, random $k$-SAT constitutes a random NP-complete problem~\cite{Karp1972, Livne2010}, since the model naturally captures the distributional structure of random instances and supports reductions from other random NP problems.

The $k$-SAT problem asks whether a Boolean formula is satisfiable, and remains NP-complete for $k\ge3$. Its optimization counterpart, max-$k$-SAT, seeks an assignment maximizing the number of satisfied clauses and is NP-hard for $k\ge2$. In this paper, we instead consider a restricted optimization variant, max-$k$-SSAT, which only involves satisfiable instances. Although no known polynomial-time reduction from max-$k$-SSAT to an established NP-complete problem currently exists, Lemma~\ref{lem_problem_reduction} shows that an efficient algorithm for max-$k$-SSAT would imply a polynomial-time algorithm for the $k$-SAT decision problem. Consequently, the existence of such an algorithm would have direct implications for the P versus NP question.

To characterize random max-$k$-SSAT instances, we introduce two random models. The first model, $F_s(n,m,k)$, extends the standard $F(n,m,k)$ model by selectively generating clauses while preserving satisfiability; unsatisfiable branches are discarded during the construction process. The second model, $F_f(n,m,k)$, serves as a theoretically motivated approximation to $F_s(n,m,k)$. In this model, a random assignment $t_0$ is first selected, and only clauses satisfied by $t_0$ are sampled. The corresponding clause-generation procedures are illustrated in Fig.~\ref{fig_clause}.

Although $F_s(n,m,k)$ provides a more natural representation of satisfiable random instances, the model $F_f(n,m,k)$ is substantially more amenable to theoretical analysis. Despite their structural differences, Lemma~\ref{lem_model_reduction} shows that the two models possess similar computational properties, making $F_f(n,m,k)$ a practical analytical framework for studying random satisfiable instances.

\begin{figure}[!t]
	\centering
	\begin{minipage}[c]{0.45\textwidth}
		{\includegraphics[width=1\linewidth]{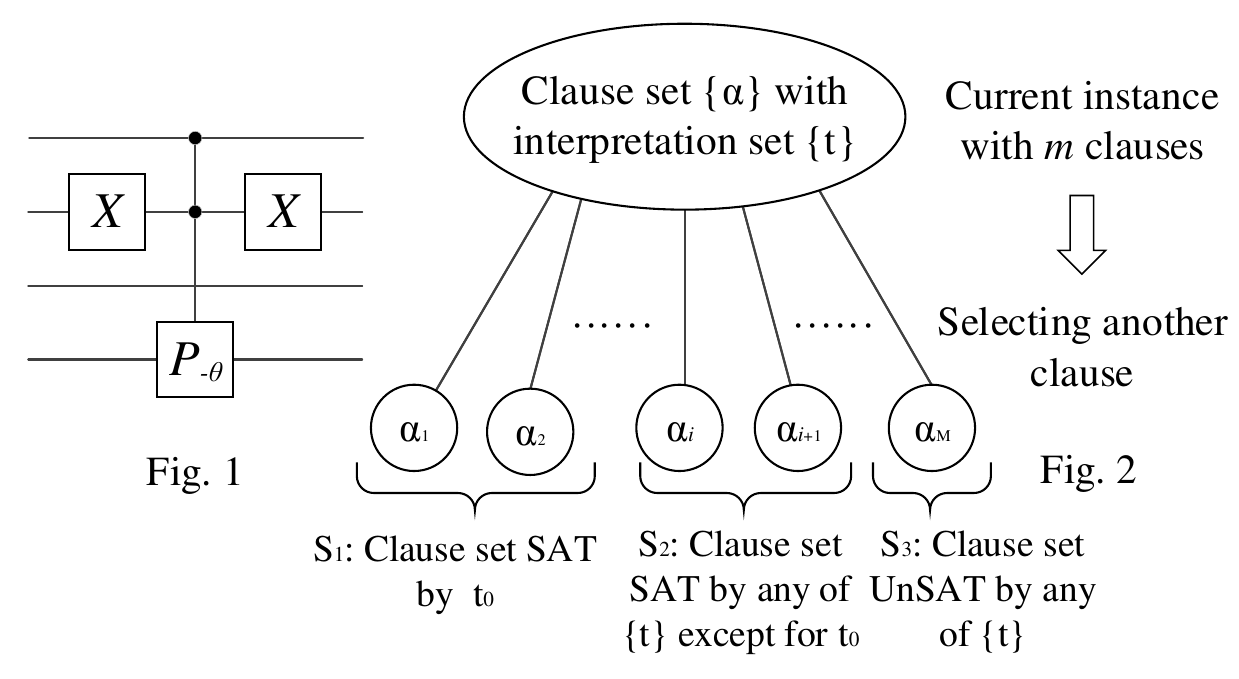}}
	\end{minipage}%\begin{minipage}[c]{0.17\textwidth}
	\footnotesize
	\caption{An illustrative example of quantum circuit for $e^{i\theta H_\alpha}$, where $\alpha = \lnot x_1 \lor x_2 \lor \lnot x_4$. This clause is represented as $\alpha = (-1, 2, -4)$, corresponding to the Hamiltonian term $H_\alpha = -p_{\bar\alpha}$, where $\bar\alpha=(1, -2, 4)$. This implies that the 1st, 2nd, and 4th qubits (from top to bottom in the figure) are occupied, and an additional pair of $X$ gate is applied to the 2nd qubit. In this circuit, the 4th qubit serves as the controlled qubit; however, it is notable that the controlled qubit can be any of the involved qubits, with the rest qubits acting as the control qubits.}
	\label{fig_U1gate}
	\footnotesize
	\caption{Clause selection process in generating a random $k$-SAT instance. Given an instance $I \in U_s$ with satisfying assignment set $\{t\}$, model $F(n, m, k)$ selects clauses uniformly at random from the full set $S_1 \cup S_2 \cup S_3$. In contrast, $F_s(n, m, k)$ samples only from clauses that are satisfied by at least one $t \in \{t\}$, i.e., from $S_1 \cup S_2$, while $F_f(n, m, k)$ restricts selection to clauses in $S_1$ that are satisfied by a fixed assignment $t_0$. When ${t} = {t_0}$, the selection sets of $F_s(n, m, k)$ and $F_f(n, m, k)$ coincide, and as $m$ increases, $F_s(n, m, k)$ gradually converges to $F_f(n, m, k)$.}
	\label{fig_clause}
\end{figure}

\subsection{Problem Hamiltonian of k-SAT instances}\label{apsubsec_HC_decomposation}

In this paper, we use multi-controlled phase gates to implement the evolution operator of the problem Hamiltonian associated with a Boolean formula. Consider a Boolean conjunction
\[
\alpha = (\lnot) x_{a_1} \land (\lnot)x_{a_2} \land \cdots \land (\lnot)x_{a_k},
\]
whose characteristic function satisfies $f_\alpha(x)=1$ if and only if the conjunction is satisfied, and $f_\alpha(x)=0$ otherwise. The corresponding problem Hamiltonian $h_\alpha$ is therefore defined through
\[
h_\alpha \ket{x}= f_\alpha(x) \ket{x}.
\]
We denote the conjunction compactly as $\alpha=(\pm a_1,\pm a_2,\cdots,\pm a_k)$, where $+a_t$ corresponds to the literal $x_{a_t}$ and $-a_t$ corresponds to $\lnot x_{a_t}$. Under this notation, the problem Hamiltonian $h_\alpha$ can be written as
\begin{equation}
	{{h}_{\alpha}}=\sigma_{z^\mp}^{(a_1)} \otimes \sigma_{z^\mp}^{(a_2)} \otimes \cdots \otimes \sigma_{z^\mp}^{(a_k)},
\end{equation}
where a positive (negative) literal corresponds to $\sigma_{z^-}$ ($\sigma_{z^+}$), respectively, while all remaining qubit positions are tensored with the identity operator $I$. Here, $\sigma_{z^\pm}^{(a_j)}$ denotes $\sigma_{z^\pm}$ acting on the $a_j$-th qubit, and
\[
\sigma_{z^\pm} = \frac{1}{2}(I\pm\sigma_z)
\]
are the projectors onto $\ket{0}$ and $\ket{1}$, respectively, with matrix representations
\begin{equation}
	\sigma_{z^+} = \begin{bmatrix} 1 & 0 \\ 0 & 0 \end{bmatrix}, \quad
	\sigma_{z^-} = \begin{bmatrix} 0 & 0 \\ 0 & 1 \end{bmatrix}.
\end{equation}

The corresponding evolution operator $e^{i\theta h_\alpha}$ is equivalent to a $(k-1)$-controlled phase gate, which can be implemented with gate complexity ${\mathcal O}(k)$~\cite{Barenco1995}. In the single-qubit case, this operator reduces to the phase gate
\[
P_\theta=e^{i\theta\sigma_{z^-}},
\]
whose matrix form is
\begin{equation}
	P_\theta = \begin{bmatrix} 1 & 0 \\ 0 & e^{i\theta} \end{bmatrix}.
\end{equation}
An example quantum circuit corresponding to a specific Boolean formula is shown in Fig.~\ref{fig_U1gate}.

In the $k$-SAT problem, each clause $\alpha$ is a Boolean disjunction of the form
\[
\alpha = (\lnot) x_{a_1} \lor (\lnot)x_{a_2} \lor \cdots \lor (\lnot)x_{a_k},
\]
which is similarly denoted as $(\pm a_1,\pm a_2,\ldots, \pm a_k)$. Letting the complement of $\alpha$ be
\[
\bar\alpha = (\mp a_1,\mp a_2,\ldots, \mp a_k),
\]
the corresponding clause Hamiltonian can be expressed as
\[
H_\alpha = -h_{\bar\alpha}.
\]
Since the objective function of the Boolean formula is given by
\[
f(x) = \sum_\alpha f_\alpha(x),
\]
the total problem Hamiltonian takes the form
\begin{equation}
	H_C = \sum_\alpha H_\alpha.
\end{equation}
Because $H_C$ is diagonal in the computational basis, its evolution operator can be decomposed into the product of the evolutions generated by the individual clause Hamiltonians:
\begin{equation}
	e^{-i \theta H_C} = \prod_\alpha e^{-i \theta H_\alpha}.
\end{equation}

\section{Proof of efficiency of the universal-mixer \textit{k}-local quantum search}
\label{apsec_proof_universal_mixer}

In this section, we provide a complete proof of the efficiency of the universal-mixer variant of the \textit{k}-local quantum search introduced in Sec.~\ref{subsec_connecting_qaoa}. Our analysis builds upon the theoretical framework developed in Ref.~\cite{mine1}, where the original \textit{k}-local quantum search and its adiabatic variant were established. However, the present setting introduces a crucial modification: the specialized $k$-local mixer Hamiltonian $H_{B,k}$ is replaced by the universal transverse-field mixer $H_B = \sum_j \sigma_j^x$. The main technical challenge is therefore to show that this substitution preserves the algorithmic efficiency, despite the loss of explicit locality alignment between the problem and mixer Hamiltonians.

The remainder of this section is organized as follows. Sec.~\ref{apsubsec_proof_circuit} revisits the circuit constructions, introduces the notation used throughout the proofs, and collects several preliminary results adapted from Ref.~\cite{mine1}; Sec.~\ref{apsubsec_proof_search_on_search} establishes that the universal-mixer circuit retains the efficiency proven for the idealized \textit{k}-local search problem; Sec.~\ref{apsubsec_proof_search_on_SAT} extends this result to random max-$k$-SSAT instances via the statistical correspondence described in the main text; and finally, Sec.~\ref{apsubsec_proof_adiabatic_on_SAT} analyzes the adiabatic variant, establishes the spectral gap bound, and derives the corresponding average-case complexity.

Throughout this section, we emphasize the similarities and differences with the original analysis in Ref.~\cite{mine1}, and explicitly highlight the new arguments required for handling the universal mixer.

% 先重述线路，及其分解以及一些证明的前置
\subsection{Framework and basic properties of \textit{k}-local quantum search} 
\label{apsubsec_proof_circuit}

We begin by reviewing the circuit construction of the \textit{k}-local quantum search introduced in Ref.~\cite{mine1}, and adapting it to the universal-mixer setting considered in this work. 
This subsection also collects several preliminary results that will be used throughout the subsequent analysis.

% 第一部分解释介绍线路
\subsubsection{Circuit formulation of original \textit{k}-local quantum search}
\label{apsubsubsec_circuit_k_local_QS}

The \textit{k}-local quantum search algorithm is defined through the iterative application of a problem-dependent phase operator and a structured mixer operator. 
Starting from the uniform superposition $|+\rangle^{\otimes n}$, the state evolves according to
\begin{equation}\label{eq_ap_k_local_QS}
	|\psi\rangle = \left( e^{-i\theta H_{B,k}} e^{-i\theta H_k} \right)^p |+\rangle^{\otimes n},
\end{equation}
where $H_k$ is the problem Hamiltonian encoding the objective function $f_k(x)$ via
\begin{equation}
	H_k |x\rangle = f_k(x) |x\rangle,
\end{equation}
and $H_{B,k}$ is a $k$-local mixer Hamiltonian designed to match the locality structure of $H_k$.

To implement the evolution operator $e^{-i\theta H_k}$, the Hamiltonian $H_k$ can be decomposed into a sum of $k$-local components, as shown in Ref.~\cite{mine1}. 
Consider first the reference Hamiltonian $H_{k,0}$ corresponding to the target string $t=0$. 
Let $I_{n,k}$ denote the set of all $k$-tuples of qubit indices. 
For each $\alpha \in I_{n,k}$, define the local operator $h_k^{(\alpha)}$ as $h_k = \sigma_{z^+}^{\otimes k}$ acting on the qubits indexed by $\alpha$. 
Then $H_{k,0}$ admits the decomposition
\begin{equation}
	H_{k,0} = \frac{1}{C_n^k} \sum_{\alpha \in I_{n,k}} h_k^{(\alpha)}.
\end{equation}
Since $H_{k,0}$ is diagonal, its evolution operator factorizes as
\begin{equation}
	e^{-i\theta H_{k,0}} = \prod_{\alpha \in I_{n,k}} e^{-\frac{i\theta}{C_n^k} h_k^{(\alpha)}}.
\end{equation}
For a general target string $t$, the Hamiltonian $H_k$ differs from $H_{k,0}$ only by local basis flips, which can be implemented by inserting appropriate $X$ gates before and after each local evolution.

The mixer Hamiltonian $H_{B,k}$ is defined as a $k$-local generalization of the transverse-field operator, given by
\begin{equation}
	H_{B,k} = \frac{1}{C_n^k} \sum_{\alpha \in I_{n,k}} X_k^{(\alpha)},
\end{equation}
where $X_k = H^{\otimes k} h_k H^{\otimes k}$, and $X_k^{(\alpha)}$ acts on the qubits indexed by $\alpha$. 
Equivalently, $H_{B,k}$ can be written as
\begin{equation}
	H_{B,k} = H^{\otimes n} H_{k,0} H^{\otimes n},
\end{equation}
which follows from a basis transformation argument (see Ref.~\cite{mine1}). 
When $k=1$, the operator $H_{B,k}$ reduces to the standard transverse-field Hamiltonian $H_B = \sum_j \sigma_j^x$, up to a normalization factor and a global phase.

Under the above decomposition, both $e^{-i\theta H_k}$ and $e^{-i\theta H_{B,k}}$ can be implemented as products of local unitary operations. 
Each term $e^{-i\theta h_k^{(\alpha)}}$ corresponds to a $(k-1)$-controlled phase gate, which can be realized with gate complexity $\mathcal{O}(k)$~\cite{Barenco1995}. 
Therefore, for a small constant $k$, a single iteration of the unitary
\begin{equation}
	U_k = e^{-i\theta H_{B,k}} e^{-i\theta H_k}
\end{equation}
can be implemented efficiently using $\mathcal{O}(n^k)$ elementary quantum gates.

The above algorithm was originally introduced to solve the \textit{k}-local search problem, where the objective function $f_k(x)$ is defined through local constraints. As shown in Ref.~\cite{mine1}, for small constant $k$, with $\theta = \Theta(n^{-1})$, this algorithm achieves efficient amplification with a number of iterations scaling as $p = \mathcal{O}(n^2)$. 

When applied to random max-$k$-SSAT instances, the ideal Hamiltonian $H_k$ is replaced by the normalized problem Hamiltonian $\bar{H}_C$, leading to the circuit
\begin{equation}
	|\psi\rangle = \left( e^{-i\theta H_{B,k}} e^{-i\theta \bar{H}_C} \right)^{p_\theta} |+\rangle^{\otimes n}.
\end{equation}
As established in Ref.~\cite{mine1}, the spectrum of $\bar{H}_C$ concentrates around that of $H_k$, and the deviation $\Delta H_C = \bar{H}_C - H_k$ vanishes in probability as the clause density $m \to \infty$. 
It was shown in Ref.~\cite{mine1} that, for any fixed $\epsilon>0$ and sufficiently large $n$, when $m=\Theta(n^{2+\epsilon+\delta})$, the $k$-local quantum search algorithm retains the same efficiency scaling as in the ideal $k$-local search setting with probability $1-\mathcal{O}(\mathrm{erfc}(n^{\delta/2}))$.

% 接着介绍绝热的部分内容
\subsubsection{Adiabatic formulation of \textit{k}-local quantum search}
\label{apsubsubsec_circuit_adiabatic_k_local_QS}

While the circuit-based implementation relies on iterative applications of local unitaries, its performance may become sensitive to deviations from the ideal Hamiltonian when such perturbations accumulate over many iterations. 
To address this limitation, an adiabatic variant of the \textit{k}-local quantum search was introduced in Ref.~\cite{mine1}. Specifically, the time-dependent system Hamiltonian is defined as
\begin{equation}
	\bar{H}_k(s) = s H_k + (1-s) H_{B,k},
\end{equation}
where $s \in [0,1]$ is a continuous interpolation parameter. 
At $s=0$, the system Hamiltonian reduces to $H_{B,k}$, whose highest-energy eigenstate is the uniform superposition $|+\rangle^{\otimes n}$. As $s$ increases, the Hamiltonian gradually incorporates the problem structure encoded in $H_k$, and at $s=1$, it coincides with the problem Hamiltonian.

Under sufficiently slow evolution, the adiabatic theorem ensures that the system remains in its instantaneous eigenstate throughout the process. 
In the present setting, the system is initialized in the highest-energy eigenstate of $H_{B,k}$ and evolves along the corresponding eigenstate branch of $\bar{H}_k(s)$, eventually reaching the target state that maximizes the objective function $f_k(x)$.

Ref.~\cite{mine1} considers a linear interpolation schedule $s = t/T$, leading to the time-dependent Hamiltonian
\begin{equation}
	\bar{H}_k(t) = \frac{t}{T} H_k + \left(1-\frac{t}{T}\right) H_{B,k}.
\end{equation}
The efficiency of the adiabatic evolution is governed by the spectral properties of $\bar{H}_k(s)$, in particular the minimum spectral gap along the evolution path. As shown in Ref.~\cite{mine1}, for the $k$-local search problem, the minimum spectral gap scales as $g_{k,0} = \Theta(n^{-1})$, leading to an overall runtime $T = \mathcal{O}(n^2)$.

When applied to random max-$k$-SSAT instances, the ideal Hamiltonian $H_k$ is replaced by the normalized problem Hamiltonian $\bar{H}_C$, resulting in the perturbed interpolation
\begin{equation}
	\bar{H}(s) = s \bar{H}_C + (1-s) H_{B,k}.
\end{equation}
Due to the enhanced robustness of adiabatic evolution against perturbations, the required clause density for efficient performance can be relaxed. 
In particular, it was shown in Ref.~\cite{mine1} that, with total runtime $T = \mathcal{O}(n^2)$, the adiabatic formulation remains efficient for $m = \Theta(n^{1+\epsilon+\delta})$, achieving efficiency comparable to that of the $k$-local search problem with probability $1-\mathcal{O}(\mathrm{erfc}(n^{\delta/2}))$.

% 首先就是说明我们研究的起点，$H_{Bk}$和$H_{B}$.

% 然后给出变体的搜索

% 为了应对同样会出现的偏差的问题，我们也给出了对应绝热形式，这里你先起草一个

% 介绍新提出的线路
\subsubsection{Circuit of the universal-mixer variant}

We consider the universal-mixer variant introduced in the main text, obtained by replacing the structured mixer $H_{B,k}$ with the $1$-local mixer $H_{B,1}$. The resulting circuit is given by
\begin{equation}\label{eq_ap_um_k_local_search}
	|\psi\rangle = \left( e^{-i\theta H_{B,1}} e^{-i\theta H_k} \right)^p |+\rangle^{\otimes n}.
\end{equation}
For small constant $k$, the operator-norm difference between $H_{B,k}$ and $H_{B,1}$ remains bounded, implying a close structural proximity between the two mixers. Under this condition, the above circuit preserves the efficiency scaling of the original \textit{k}-local quantum search for the $k$-local search problem.

When extended to random max-$k$-SSAT instances, replacing $H_k$ with the normalized problem Hamiltonian $\bar{H}_C$ yields
\begin{equation}
	|\psi\rangle = \left( e^{-i\theta H_{B,1}} e^{-i\theta \bar{H}_C} \right)^p |+\rangle^{\otimes n},
\end{equation}
for which the same efficiency guarantees hold in the corresponding clause-density regime as in the structured-mixer case.

Compared with the original construction, this modification removes the explicit locality matching in the mixer operator, bringing the circuit closer to the standard QAOA ansatz with fixed, problem-independent mixers. The main technical task is to establish that the favorable spectral properties underlying the original analysis are preserved under this replacement; the corresponding proofs are provided in Appendices~\ref{apsubsec_proof_search_on_search} and~\ref{apsubsec_proof_search_on_SAT}.

We next consider the corresponding adiabatic formulation under the universal mixer. The time-dependent system Hamiltonian is defined as
\begin{equation}
	\bar{\mathcal{H}}(s) = s H_k + (1-s) H_{B,1}.
\end{equation}
For the $k$-local search problem, the minimum spectral gap along this interpolation remains of order $\Theta(n^{-1})$, which leads to a total runtime scaling as $\mathcal{O}(n^2)$ under a linear schedule. 

For random max-$k$-SSAT instances, replacing $H_k$ with $\bar{H}_C$ yields the perturbed interpolation
\begin{equation}
	\bar{\mathcal{H}}(s) = s \bar{H}_C + (1-s) H_{B,1}.
\end{equation}
The adiabatic evolution retains the same robustness against perturbations as in the structured-mixer setting, and achieves comparable efficiency in the corresponding clause-density regime, with success probability of the same asymptotic order; see Appendix~\ref{apsubsec_proof_adiabatic_on_SAT}.
% 不用讲Trotterized circuit，这是别的地方讨论的

% 为了证明我们的结论，我们借用了一些之前论文的结论，然后这里就罗列出来
\subsubsection{Preliminary results and known properties}

To facilitate the analysis in the subsequent subsections, we collect several technical results established in Ref.~\cite{mine1}, which will serve as key ingredients in our proofs. For completeness, we restate these results in a form tailored to our setting.

We start with the simplest case of $k=1$, which serves as the base case for the inductive arguments developed later. In this setting, the $k$-local quantum search circuit reduces to
\begin{equation}
	\ket{\psi} = \left( e^{-i\theta H_{B,1}} e^{-i\theta H_1} \right)^p \ket{+}^{\otimes n}.
\end{equation}
This case captures the minimal locality structure and establishes the baseline performance of the algorithm. The corresponding unitary operator can be approximately eigendecomposed locally on each qubit, as detailed in Lemma~\ref{lemma_eigen_1_local_QS}. Its effectiveness has been rigorously demonstrated in Ref.~\cite{mine1}, and is summarized in Lemma~\ref{lemma_1_local_QS}.

Next, we recall several properties that will play a central role in our analysis.  Let $h_k = \sigma_{z^+}^{\otimes k}$ and $X_k = H^{\otimes k} h_k H^{\otimes k}$. Lemmas~\ref{lemma_exchange_H}, \ref{lemma_HX_commutator} and \ref{lemma_commutator} show that the corresponding unitary operators either commute exactly or approximately commute for sufficiently small angles. Finally, we recall a model reduction result that relates different random max-$k$-SSAT ensembles, as formalized in Lemma~\ref{lem_model_reduction}. 

\begin{lem}[Lemma~B.4 in Ref.~\cite{mine1}]\label{lemma_eigen_1_local_QS}
	The 1-local search operator $U_C = H^{\otimes n} e^{i\pi H_Z/2n} H^{\otimes n} e^{i\pi H_Z/2n}$ can be approximately eigendecomposed as
	\begin{equation}
		U_C = {{V}^{T}}EV + \mathcal{O}(n^{-2}), 
	\end{equation}
	where $V=V_{0}^{\otimes n}$, and $E=E_{0}^{\otimes n}$ with 
	\begin{equation*}
		{{E}_{0}}=\left[ \begin{matrix}
			{{e}^{\frac{i\pi }{\sqrt{2}n}}} & 0  \\
			0 & {{e}^{\frac{-i\pi }{\sqrt{2}n}}}  \\
		\end{matrix} \right],
		{{V}_{0}}=\left[ \begin{matrix}
			\cos \left( \frac{\pi }{8} \right) & \sin \left( \frac{\pi }{8} \right)  \\
			-\sin \left( \frac{\pi }{8} \right) & \cos \left( \frac{\pi }{8} \right)  \\
		\end{matrix} \right].
	\end{equation*}
\end{lem}

\begin{lem}[Lemma~B.5 in Ref.~\cite{mine1}]\label{lemma_1_local_QS}
	For $k=1$, the $k$-local quantum search requires approximately $\frac{n}{\sqrt{2}}$ oracle calls to evolve the initial state toward the target state $\ket{t}$. Moreover, the amplitude on $\ket{t}$ converges to $1$ as $n$ increases.
\end{lem}

\begin{lem}[Lemma~B.1 in Ref.~\cite{mine1}]\label{lemma_exchange_H}
	Given a single-qubit gate $M$ that satisfies $M^2=I$, $\mathcal{H}_{M,k} = M^{\otimes n} (\sum_{\alpha \in I_{n,k}} h_\alpha)M^{\otimes n}$ can equivalently be written as $H_{M,k} =\sum_{\alpha \in I_{n,k}} M_k^{(\alpha)}$, where $M_k^{(\alpha)}$ denotes $M_k$ acting on qubits indexed by $\alpha$, and $M_k = M^{\otimes k} h_k M^{\otimes k}$.
\end{lem}

\begin{lem}[Lemma~B.3 in Ref.~\cite{mine1}]\label{lemma_HX_commutator}
	For any integer $k$, the operator $H^k e^{i\theta h_k} H^k$ commutes with the $X$ gate, irrespective of upon which qubit the $X$ gate operates.
\end{lem}

\begin{lem}[Lemma~B.7 in Ref.~\cite{mine1}]\label{lemma_commutator}
	For any sufficiently small $\theta_1,\theta_2$, if the qubits acted upon by $e^{i\theta_1 h_{k_1}}$ overlap with those acted upon by $e^{i\theta_2 X_{k_2}}$, then the two operators approximately commute, with an error of order ${\mathcal O}(\theta_1 \theta_2)$.
\end{lem}

\begin{lem}[Lemma~3.3 in Ref.~\cite{mine1}]\label{lem_model_reduction}
	For any fixed $\epsilon > 0$ and sufficiently large $n$, if $m = \Omega(n^{1+\epsilon})$ and the random max-$k$-SSAT problem under $F_f(n,m,k)$ can be solved in polynomial time on average, then the corresponding problem under $F_s(n,m,k)$ can also be solved in polynomial time on average.
\end{lem}

In addition to the results presented above, a slight modification of a lemma from Ref.~\cite{mine1} is required to accommodate arbitrary target states $t$ in the circuit. The modified statement and its corresponding proof are given in Lemma~\ref{lemma_reduce_t}. This lemma establishes that, after generalizing the locality of the circuit, it suffices to consider the case $t = 0$ in the analysis of an arbitrary target state.

\begin{lem}\label{lemma_reduce_t}
	Consider a $k$-local quantum search circuit of the form
	\begin{equation}\label{eq_general_k_local_QS}
		\ket{\psi} = \left( e^{-i\theta H_{B,k_0}} e^{-i\theta H_{k}} \right)^p \ket{+}^{\otimes n}.
	\end{equation}
	The evolution with an arbitrary target state $t$ can be reduced to the case with $t = 0$.
\end{lem}
\begin{proof}
	% reduce the circuit to the standard form 
	% 首先是k = 1
	We first consider the case $k = k_0 = 1$. The Hamiltonian $H_{k,0}$ is proportional to $H_Z = \sum_j \sigma_z^{(j)}$, i.e., $H_{k,0}$ equals $H_Z / 2n$ up to a global phase. Using the identity $e^{-i\theta \sigma_z} = X e^{i\theta \sigma_z} X$, the Hamiltonian $H_{k}$ for an arbitrary target $t$ can be expressed as $H_{k,0}$ conjugated by a set of $X^{t_j}$ gates acting on the $j$-th qubit from both sides, where $t_j$ is the $j$-th bit of $t$, with $X^1 = X$ and $X^0 = I$. 
	
	Consequently, the evolution of the 1-local quantum search reads
	\begin{equation*}
		\ket{\psi_p} = \left( H^{\otimes n} e^{-i \pi H_Z / 2n} H^{\otimes n} X_C e^{-i \pi H_Z / 2n} X_C \right)^p \ket{+}^{\otimes n},
	\end{equation*}
	where $X_C = \prod_j X_{(j)}^{t_j}$, and $X_{(j)}^{t_j}$ denotes $X^{t_j}$ acting on the $j$-th qubit. 
	Since $H e^{i\theta \sigma_z} H$ commutes with $X^{t_j}$, most of the $X_C$ operators cancel, leaving only a single $X_C$ at the front. 
	Thus, the evolution reduces to
	\begin{equation*}
		\ket{\psi_p} = X_C \left( H^{\otimes n} e^{i \pi H_Z / 2n} H^{\otimes n} e^{i \pi H_Z / 2n} \right)^p \ket{+}^{\otimes n}.
	\end{equation*}
	Here, the operator $X_C$ fully encodes the target state $\ket{t}$, since $X_C \ket{0}^{\otimes n} = \ket{t}$. 
	Therefore, the 1-local quantum search circuit with arbitrary target $t$ reduces to the standard case with $t = 0$.  
	
	We now extend the argument to general $k$ and $k_0$. 
	Using $X_{(j)}^{t_j} X_{(j)}^{t_j} = I$, any additional pairs of $X_{(j)}^{t_j}$ gates appearing in each term $e^{i\theta h_k^{(\alpha)}}$ can be factored out by Lemma~\ref{lemma_exchange_H}. 
	Accordingly, the evolution with target $t$ can be written as
	\begin{equation*}
		\ket{\psi_p} = \left( H^{\otimes n} e^{-i\pi H_{k_0}} H^{\otimes n} X_C e^{-i\pi H_{k}} X_C \right)^p \ket{+}^{\otimes n}.
	\end{equation*}
	Moreover, Lemma~\ref{lemma_HX_commutator} ensures that $H^k e^{i\theta h_k} H^k$ commutes with $X$ gates on all qubits. By the same reasoning as above, the reduction to the standard form extends naturally to the $k$-local quantum search circuit.
\end{proof}

\subsection{Proof of efficiency of universal-mixer variant on search problems}
\label{apsubsec_proof_search_on_search}
% 先证明在搜索问题上的有效性

This subsection analyzes the efficiency of the universal-mixer $k$-local quantum search in the regime of small constant $k$, extending the result established in Lemma~\ref{lemma_1_local_QS} for the 1-local case. As $k$ increases, the objective function of the $k$-local search evolves gradually. To quantify this change, we define the deviation
\begin{equation}
	\Delta f_{k}(x) = f_{k+1}(x) - f_{k}(x),
\end{equation}
which can be explicitly expressed as
\begin{equation}\label{eq_deviation_fkx}
	\Delta f_{k}(x) = -\frac{n-l}{n-k} f_k(x),
\end{equation}
where $l = n - d_H(x,t)$ and $d_H(x,t)$ denotes the Hamming distance between $x$ and the target $t$. Importantly, $\Delta f_{k}(x) \le 0$ for any $k<n$ and arbitrary $x$, reflecting that the increase in locality reduces the contribution of structural information inherent in the objective function.

We begin by recalling the inductive framework developed in Ref.~\cite{mine1}, where the efficiency of the \textit{k}-local quantum search is established via an induction on the locality parameter $k$. Starting from the $k=1$ case (Lemma~\ref{lemma_1_local_QS}), both the problem Hamiltonian and the mixer Hamiltonian are lifted to higher locality in a coordinated manner, leading to the circuit in Eq.~\eqref{eq_ap_k_local_QS}.

In contrast, the universal-mixer variant considered in this work follows a different inductive strategy. Instead of increasing the locality of both components, we fix the mixer Hamiltonian at $k=1$, as in Eq.~\eqref{eq_ap_um_k_local_search}, and only increase the locality of the problem Hamiltonian. The analysis is therefore carried out by induction on $k$, starting from the $k=1$ case, while treating the increment from $H_k$ to $H_{k+1}$ explicitly. We focus on the regime of small constant $k$, and show that this inductive step preserves the efficiency scaling of the algorithm.

For the discrete circuit in Eq.~\eqref{eq_ap_um_k_local_search}, directly analyzing the fidelity of the evolved state is technically challenging. 
Following the approach of Ref.~\cite{mine1}, we instead interpret the evolution from the perspective of phase rotations. Applying the second-order Trotter--Suzuki formula
\begin{equation*}
	e^{-i\theta (H_A + H_B)} 
	= e^{-i\theta/2 H_A} e^{-i\theta H_B} e^{-i\theta/2 H_A} + \mathcal{O}(\theta^3),
\end{equation*}
the unitary evolution can be approximated as
\begin{align}\label{eq_trotterized_um_QS}
	\left( e^{-i\theta H_{B,n,1}} e^{-i\theta H_{n,k}} \right)^{p} |+\rangle
	&= e^{-i p \theta \mathcal{H}_{n,k,1}} |+\rangle \nonumber\\
	&\quad + \mathcal{O}(n^{-1}) + \mathcal{O}(\theta^3 p)
\end{align}
where the effective Hamiltonian is defined as
\begin{equation}
	\mathcal{H}_{n,k,1} = H_{n,k} + H_{B,n,1}.
\end{equation}
Here $H_{n,k}$ and $H_{B,n,1}$ are the scaled versions of $H_{k}$ and $H_{B,1}$, respectively, with an additional factor of order $n$ introduced.

Under this approximation, the discrete-time evolution is reduced to a continuous-time evolution generated by $\mathcal{H}_k$, in which the dynamics can be viewed as rotations within its eigenspace. Consequently, the performance of the algorithm is governed by the spectral properties of $\mathcal{H}_k$. Specifically, the key quantity is the spectral gap $g_k$, defined as the energy difference between the highest and second-highest eigenvalues of $\mathcal{H}_k$. Here the target state corresponds to the computational basis state with the highest energy under the problem Hamiltonian. Under this formulation, the central task is to establish that
\begin{equation*}
	g_k = \Theta(n^{-1}).
\end{equation*}

The proof follows the same structural route as in Ref.~\cite{mine1}. Lemma~\ref{lemma_eigen_reduction} characterize the algebraic structure of the Hamiltonian and its invariant subspace. Based on these results, Lemma~\ref{lemma_gap} establishes that the spectral gap of $\mathcal{H}_k$ scales as $\Theta(n^{-1})$. Finally, Lemma~\ref{lemma_k_local_QS} implies that, for $\theta = \Theta(n^{-1})$, the required number of iterations satisfies
\begin{equation*}
	p = \mathcal{O}(n^2),
\end{equation*}
which completes the proof of efficiency.

\begin{lem}\label{lemma_eigen_reduction}
	For a sufficiently large $n$, the eigenvalues of $\mathcal{H}_{n+1,k}$ are of the same order of magnitude as those of $\tilde{\mathcal{H}}_{n+1,k}$, where
	\begin{equation}\label{eq_approxi_Hn}
		\tilde{\mathcal{H}}_{n+1,k} = \left[ \begin{matrix}
			\frac{n+1-k}{n+1}\mathcal{H}_{n,k} &  \\
			& 	\frac{n+1-k}{n+1}\mathcal{H}_{n,k} + \frac{k}{n+1}(H_{n,k-1}+I)  \\
		\end{matrix} \right].
	\end{equation}
\end{lem}
\begin{proof}
	% give the decomposition and reduce the operator 
	With $\theta = \Theta(n^{-1})$, the analysis of the spectral gap of $\mathcal{H}_{n,k}$ can be simplified by considering the gap of the Hamiltonian associated with 
	\begin{equation*}
		U_{n+1,k}(\theta) = e^{-i\theta H_{B,n+1,1}} e^{-i\theta H_{n+1,k}}.
	\end{equation*}
	Moreover, the Hamiltonian 
	\begin{equation*}
		H_{n+1,k} = \frac{1}{C_{n+1}^k} \sum_{\alpha \in I_{n+1,k}} h_k^{(\alpha)}
	\end{equation*}
	can be decomposed according to whether the ($n$+1)-th qubit is involved in the evolution. The only difference lies in the additional normalization factor $1/C_{n+1}^k$. Explicitly, we have
	\begin{equation*}
		H_{n+1,k} = \frac{n+1-k}{n+1} I \otimes H_{n,k} + \frac{k}{n+1} H'_{n,k-1},
	\end{equation*}
	where $H'_{n,k-1}$ corresponds to $H_{n,k-1}$ with an additional control qubit at the ($n$+1)-th position.  
	
	For the mixer Hamiltonian $H_{B,n+1,1}$, which acts locally on each qubit, we can decompose it into two parts as follows:
	\begin{equation*}
		H_{B,n+1,1} = \frac{n}{n+1} I \otimes H_{B,n,1} + \frac{1}{n+1} H_{B,1,1} \otimes I_n,
	\end{equation*}
	where $I_n$ denotes the identity operator on the first $n$ qubits. In the following, we denote by $H_{B,1,1}^{(n+1)}$ the single-qubit mixer Hamiltonian $H_{B,1,1}$ acting on the ($n$+1)-th qubit; explicitly, $H_{B,1,1}^{(n+1)} = H h_1 H$.
	
	Based on this, $U_{n+1,k}(\theta)$ can be expanded as 
	\begin{align*}
		U_{n+1,k}(\theta) & = e^{-i\theta\frac{n}{n+1} H_{B,n,1}} e^{-i\theta\frac{1}{n+1} H_{B,1,1}^{(n+1)}}  \nonumber\\
		& \quad \times  e^{-i\theta\frac{k}{n+1} H'_{n,k-1}} e^{-i\theta\frac{n+1-k}{n+1} H_{n,k}}.   
	\end{align*}
	There are $C_n^{k-1}$ terms in $e^{-i\theta\frac{k}{n+1} H'_{n,k-1}}$, each associated with a coefficient $1/C_{n}^{k-1}$. Due to the factors, $e^{-i\theta\frac{k}{n+1} H'_{n,k-1}}$ can commute with $e^{-i\theta\frac{1}{n+1} H_{B,1,1}}$ with a deviation on the order of ${\mathcal O}(\theta^2 n^{-2})$, as established in Lemma \ref{lemma_commutator}. Consequently, the evolution is reduced to 
	\begin{align*}
		U_{n+1,k}(\theta) =  & \, e^{-i\theta\frac{n}{n+1} H_{B,n,1}}  e^{-i\theta\frac{k}{n+1} H'_{n,k-1}} \\\nonumber
		& e^{-i\theta\frac{1}{n+1} H^{(n+1)}_{B,1,1}} e^{-i\theta\frac{n+1-k}{n+1} H_{n,k}} 
		+ {\mathcal O}(n^{-4}).
	\end{align*}
	
	% given the evolution in block matrix form
	Regardless of the negligible deviation in the operator, we denote the dominant component of $U_{n+1,k}(\theta)$ as $U_0(\theta)$. $U_0(\theta)$ is constituted by two analogous components, that is, $U_0(\theta) = U_{H_1, n+1, k}(\theta)U_{H_2, n+1, k}(\theta)$, respectively expressed as 
	\begin{align*}
		& U_{H_1, n+1, k}(\theta) = H^{\otimes n+1} e^{-i\theta\frac{n}{n+1} H_{n,1} }H^{\otimes n+1} e^{- i\theta\frac{k}{n+1} H'_{n,k-1}}, \nonumber\\
		&  U_{H_2, n+1, k}(\theta) = H^{\otimes n+1} e^{-i\theta\frac{1}{n+1} h_{1}^{(n+1)}} H^{\otimes n+1} e^{-i\theta\frac{n+1-k}{n+1} H_{n,k}}.
	\end{align*}
	% give the decompostion by the (n+1)-th qubit
	For an arbitrary computational basis $\left| x \right>$, the evolution of $U_{H_1, n+1, k}(\theta)$ on the states $\left| 0 \right> \left| x \right>$ and $\left| 1 \right> \left| x \right>$ are expressed as 
	\begin{align*}
		& U_{H_1, n+1, k}(\theta)  \left| 0 \right> \left| x \right> = e^{-i\theta\frac{n}{n+1} H_{B, n, 1} }  \left| 0 \right> \left| x \right>, \nonumber\\
		& U_{H_1, n+1, k}(\theta)  \left| 1 \right> \left| x \right> = e^{-i\theta\frac{n}{n+1} H_{B, n, 1} } e^{- i\theta\frac{k}{n+1} H_{n,k-1}} \left| 1 \right> \left| x \right>. 
	\end{align*}
	Thus, the evolution operator $U_{H_1, n+1, k}(\theta)$ can be reformulated in block-matrix form as
	\begin{align*}
	& U_{H_1, n+1, k}(\theta) = \\\nonumber
	& \quad\quad \left[ \begin{matrix}
			e^{-i\theta\frac{n+1-k}{n+1} H_{B, n, 1} } &  \\
			& e^{-i\theta\frac{n+1-k}{n+1} H_{B, n, 1} } e^{-i\theta\frac{k}{n+1} H_{n,k-1}}  \\
		\end{matrix} \right]. 
	\end{align*}
	
	Similarly, the evolution of $H^{(n+1)} U_{H_2, n+1, k}(\theta)H^{(n+1)}$ can be expressed as 
	\begin{align*}
		H^{(n+1)} U_{H_2, n+1, k}(\theta) H^{(n+1)} & \left| 0 \right> \left| x \right> = e^{-i\theta\frac{n+1-k}{n+1} H_{n,k} }  \left| 0 \right> \left| x \right>,\nonumber\\
		H^{(n+1)} U_{H_2, n+1, k}(\theta) H^{(n+1)}  &  \left| 1 \right> \left| x \right> = \nonumber\\
		& \,\,\,\, e^{- i\theta\frac{k}{n+1}} e^{-i\theta\frac{n+1-k}{n+1} H_{n,k} }  \left| 1 \right> \left| x \right>,
	\end{align*}
	which can also be represented in block-matrix form as
	\begin{align*}
		H^{(n+1)}& U_{H_2, n+1, k}(\theta)H^{(n+1)} = \nonumber\\
		&\quad\quad\quad \left[ \begin{matrix}
			e^{-i\theta\frac{n+1-k}{n+1} H_{n,k} } &  \\
			& e^{-i\theta \frac{k}{n+1}} e^{-i\theta\frac{n+1-k}{n+1}H_{n,k} } \\
		\end{matrix} \right]. 
	\end{align*}
	
	As $n$ increases, the extra pair of Hadamard gates $H^{(n+1)}$ applied to a single qubit cannot affect the magnitude of the eigenspectrum of the entire Hamiltonian. This is straightforward in a large-scale quantum system that the application of a single-qubit gate does not significantly alter the eigenspectrum of the overall Hamiltonian, as the system's energy levels are predominantly governed by the multi-qubit interactions. Consequently, our analysis focuses on the eigendecomposition of $U'_0 = U_{H_1, n+1, k}(\theta)H^{(n+1)}U_{H_2, n+1, k}(\theta)H^{(n+1)}$, which can be represented in block-matrix form as 
	\begin{equation*}
		U'_0 = \left[ \begin{matrix}
			U_{00}  &  \\
			& U_{11}  \\
		\end{matrix} \right],
	\end{equation*}
	where 
	\begin{align*}
		& U_{00} = e^{-i\theta\frac{n+1-k}{n+1} H_{B, n, 1} } e^{-i\theta\frac{n+1-k}{n+1} H_{n,k} },  \nonumber\\
		& U_{11} = e^{-i\theta\frac{n+1-k}{n+1}  H_{B, n, 1} } e^{-i\theta \frac{k}{n+1}  H_{n,k-1}} e^{-i \theta\frac{k}{n+1}}  e^{-i\theta\frac{n+1-k}{n+1}H_{n,k} }. 
	\end{align*}
	According to the Trotter-Suzuki decomposition, $U'_0$ has an identical eigenspectrum to that of 
	\begin{equation*}
		\mathcal{U}'_0 = \left[ \begin{matrix}
			e^{-i\theta\frac{n+1-k}{n+1}\mathcal{H}_{n,k} }&  \\
			& 	e^{-i\theta \left( \frac{n+1-k}{n+1}\mathcal{H}_{n,k} + \frac{k}{n+1}(H_{n,k-1}+I)  \right) }   \\
		\end{matrix} \right]
	\end{equation*}
	up to a negligible deviation. Thus, the spectral gap of the Hamiltonian associated with $U_{n+1,k}(\theta)$ is of the same magnitude as that of $\,\mathcal{U}'_0$.
\end{proof}

\begin{lem}\label{lemma_gap} 
	For a small constant $k$, the spectral gap $g_k$ of $\mathcal{H}_{n,k}$ scales as $\Theta(n^{-1})$. 
\end{lem}
\begin{proof}
	We proceed by induction on $n$. For $k=1$, Lemma~\ref{lemma_eigen_1_local_QS} establishes that the spectral gap is of order $\Theta(n^{-1})$. For a fixed small constant $k \ge 2$, the lemma can be verified directly when $n$ is not large but still sufficiently greater than $k$.
	
	Assume that the lemma holds for $\mathcal{H}_{n,k}$, i.e., its spectral gap is $\Theta(n^{-1})$. We aim to establish its validity for $\mathcal{H}_{n+1,k}$. According to Lemma~\ref{lemma_eigen_reduction}, the spectral gap of $\mathcal{H}_{n+1,k}$ is of the same order as that of the block-diagonal approximation $\tilde{\mathcal{H}}_{n+1,k}$ shown in Eq.~\eqref{eq_approxi_Hn}. The spectral gap of $\tilde{\mathcal{H}}_{n+1,k}$ is determined by the internal gaps of the blocks $H_{00}$ and $H_{11}$, as well as the energy difference between their respective maximal eigenstates.
	
	For the internal gap of $H_{00}$, since $H_{00} = \frac{n+1-k}{n+1}\mathcal{H}_{n,k}$, its spectral gap is simply the gap of $\mathcal{H}_{n,k}$ scaled by a factor close to 1. By the inductive hypothesis, this remains $\Theta(n^{-1})$. For the internal gap of $H_{11}$, utilizing the deviation of the objective function in Eq.~\eqref{eq_deviation_fkx}, we can expand the lower block as
	\begin{align*}
		H_{11} = H_{n,k} + & \frac{n+1-k}{n+1}H_{B,n,1} \nonumber\\
		& \quad + \frac{k(n-d)}{(n+1)(n-k+1)}H_{n,k-1} + \frac{k}{n+1}I.
	\end{align*}
	Note that the term involving $H_{n,k-1}$ is a diagonal positive semi-definite matrix where every element is on the order of $\mathcal{O}(n^{-2})$. Thus, to determine the magnitude of the gap, we focus on the leading-order components:
	\begin{equation*}
		H'_{11} = H_{n,k} + \frac{n+1-k}{n+1}H_{B,n,1} + \frac{k}{n+1}I.
	\end{equation*}
	Up to a global phase, consider the parameterized Hamiltonian $H(s) = H_{n,k} + sH_{B,n,1}$. Based on the reduction premise, the gap of $H(1)$ is $\Theta(n^{-1})$. Introducing a small perturbation $\Delta s = \frac{k}{n+1}$ (or a global shift) does not change the magnitude of the gap. Thus, the internal gap of $H_{11}$ remains $\Theta(n^{-1})$.
	
	Regarding the gap between the blocks, the term $\frac{k}{n+1}I$ in $H_{11}$ provides a uniform energy shift relative to $H_{00}$. The energy difference between the minimal eigenvalues of $H_{00}$ and $H_{11}$ is thus dominated by this $\Theta(n^{-1})$ shift. Combining these results, the overall spectral gap of $\tilde{\mathcal{H}}_{n+1,k}$ scales as $\Theta(n^{-1})$, completing the inductive step.
\end{proof}

Based on the spectral gap analysis above, and similarly to the original $k$-local quantum search, we can directly invoke the results from Ref.~\cite{mine1} to characterize the evolution of the universal-mixer variant. For clarity, we briefly restate the key argument here.

\begin{lem}[Lemma~B.9 in Ref.~\cite{mine1}]\label{lemma_k_local_QS}
	With $\theta = \Theta(n^{-1})$, ${\mathcal O}(n^2)$ iterations are necessary for $k$-local quantum search to evolve the state such that the amplitude of $\left| t \right>$ reaches its first local maximum.
\end{lem}
\begin{proof}
	Assume that $\mathcal{H}_{n,k}$ admits the eigendecomposition $\mathcal{H}_{n,k} = V^\dagger E V$. The evolution operator $e^{-i \theta p_\theta \mathcal{H}_{n,k}}$ can then be interpreted as a rotation within the eigenspace defined by $V$. Let $\ket{\tilde{t}}$ denote the eigenstate whose eigenvalue is closest to that of the target state $\ket{t}$. 
	
	Given that the spectral gap between $\ket{t}$ and $\ket{\tilde{t}}$ is of order $\Theta(n^{-1})$, if the number of iterations $p_\theta$ significantly exceeds $\Theta(n^2)$, the accumulated phase difference between $\ket{t}$ and $\ket{\tilde{t}}$ becomes of order $\mathcal{O}(1)$, i.e., $\theta p_\theta g_{n,k} > \mathcal{O}(1)$. Due to the periodicity of the phase, this leads to substantial over-rotation, causing phase differences among computational basis states to be effectively randomized over the interval $[0, 2\pi]$. In this regime, the amplitude of the target state can no longer increase monotonically. Consequently, the required number of iterations to reach the first local maximum is on the order of $\mathcal{O}(n^2)$.
\end{proof}

\subsection{Proof of efficiency of universal-mixer variant on max-\textit{k}-SSAT}
\label{apsubsec_proof_search_on_SAT}
% 然后证明搜索算法在SAT上的有效性

In the previous subsection, we established that the universal-mixer variant inherits the efficiency of the $k$-local quantum search algorithm on the idealized $k$-local search problem. Although this structured problem is relatively simple, it possesses a key property that allows the analysis to be extended to random max-$k$-SSAT instances. In Sec.~\ref{subsec_statistical_analysis}, our statistical study of the $F_f(n,m,k)$ ensemble revealed that the average-case behavior of the normalized problem Hamiltonian for random max-$k$-SSAT instances precisely corresponds to that of the $k$-local quantum search, as expressed in Eq.~\eqref{eq_normalized_HC}. Furthermore, the Gaussian convergence described in Eq.~\eqref{eq_normal_distribution} indicates that, in the high clause-density regime, deviations from the ideal Hamiltonian are small with high probability and do not significantly perturb its principal components. Based on these observations, this subsection demonstrates that, for random max-$k$-SSAT instances, the universal-mixer variant similarly retains its efficiency in the regime of sufficiently high clause density.

When applying the universal-mixer variant to random max-$k$-SSAT instances, the ideal $k$-local Hamiltonian $H_k$ in Eq.~\eqref{eq_ap_k_local_QS} can be replaced by the normalized problem Hamiltonian $\bar{H}_C$ for $k$-SAT instances. The resulting circuit is
\begin{equation}
	\ket{\psi} = \left( e^{-i\theta H_{B,k}} e^{-i\theta \bar{H}_C} \right)^p \ket{+}^{\otimes n}.
\end{equation}
As established in Section~\ref{subsec_statistical_analysis}, the normalized Hamiltonian $\bar{H}_C$ for random instances drawn from $F_f(n,m,k)$ converges in probability to $H_k$ as $m\to\infty$. Equivalently, the deviation $\Delta H_C = \bar{H}_C - H_k$ vanishes in probability as the clause number $m$ increases. More concretely, using the Gaussian approximation in Eq.~\eqref{eq_normal_distribution}, the eigenvalue $\mathcal{E}_{k,x}$ of $H_C$ satisfies the concentration bound
\begin{equation}\label{eq_Ekx_range}
	\left| \frac{1}{m} \mathcal{E}_{k,x} - \mu_{k,x} \right| \le \frac{c}{\sqrt{m(2^k-1)}}
\end{equation}
with probability at least $\operatorname{erf}(c/\sqrt{2})$, where $c>0$ is a tuning constant.

Based on this convergence, one can determine how large $m$ must be for the perturbation $\Delta H_C$ to remain negligible. Following the approach of Ref.~\cite{mine1}, the following theorem specifies the regime of $m$ in which, when applied to random max-$k$-SSAT instances, the universal-mixer variant achieves performance comparable to that of the ideal $k$-local search problem. For brevity, we provide only a sketch of the proof here, as the argument follows closely the reasoning in Ref.~\cite{mine1} with minor modifications to account for the universal mixer.

\begin{thm}\label{theo_threshold}
	For any fixed $\epsilon>0$ and sufficiently large $n$, the universal-mixer $k$-local quantum search algorithm, when applied to random max-$k$-SSAT instances drawn from $F_f(n,m,k)$ with $m=\Theta (n^{2+\epsilon+\delta})$, achieves efficiency comparable to that of the $k$-local search problem, with probability $1-\mathcal{O}(\mathrm{erfc}(n^{\delta/2}))$.
\end{thm}
\begin{proof}[Sketch of Proof]
	We consider the representative case $\theta = \pi$ and $p = \mathcal{O}(n)$; the argument generalizes to any $\theta < \pi$ by modifying only constant factors in the evolution operator. Setting $c = n^{\delta/2}$ and applying Eq.~\eqref{eq_Ekx_range}, the normalized eigenvalues satisfy
	\begin{equation}\label{eq_Ekx_devia}
		\left| \bar{\mathcal{E}}_{k,x} - E_{k,x} \right| \le \sqrt{\frac{2^k-1}{m}} \, n^{\delta/2},
	\end{equation}
	with probability at least $1 - \mathcal{O}(\operatorname{erfc}(n^{\delta/2}))$, where $\bar{\mathcal{E}}_{k,x}$ and $E_{k,x}$ denote the eigenvalues of $\bar{H}_C$ and $H_k$, respectively. Since $m = \Theta(n^{2+\epsilon+\delta})$, the deviation is bounded by $\mathcal{O}(n^{-(1+\epsilon/2)}) = o(n^{-1})$.
	
	The evolution under the perturbed problem Hamiltonian can be written as
	\begin{equation*}
		\ket{\psi} = \left( E_B (E_C + \Delta E_C) \right)^p \ket{+}^{\otimes n},
	\end{equation*}
	where $E_B = e^{-i\pi H_{B,1}}$, $E_C = e^{-i\pi H_k}$, and $\Delta E_C = e^{-i\pi \bar{H}_C} - e^{-i\pi H_k}$. In the computational basis, $\Delta E_C$ introduces an additional phase shift:
	\begin{equation*}
		(E_C + \Delta E_C) \sum_x \alpha_x \ket{x} = \sum_x \alpha_x e^{-i\pi (E_{k,x} + \Delta \bar{\mathcal{E}}_{k,x})} \ket{x}.
	\end{equation*}
	We write $\left| \Delta E_C \right| = o_r(n^{-1})$ to indicate that this perturbation remains of order $o(n^{-1})$ with probability $1 - \mathcal{O}(r)$, where $r = \operatorname{erfc}(n^{\delta/2})$.
	
	Let $U$ denote the full evolution operator, and define its deviation from the ideal evolution $(E_B E_C)^p$ as $\Delta U$. Since the cumulative effect of higher-order terms in $\Delta E_C$ exceeds $o(1)$ only with a probability much smaller than $\mathcal{O}(r)$, this deviation can be expressed as
	\begin{equation*}
		\Delta U = \sum_{j=0}^{p-1} \left[ (E_B E_C)^j E_B \Delta E_C (E_B E_C)^{p-j-1} \right] + o_r(1).
	\end{equation*}
	Focusing on the amplitude of the target basis state $\ket{t}$, the corresponding deviation from the ideal amplitude is
	\begin{equation*}
		\Delta P = \sum_{j=0}^{p-1} \bra{t} (E_B E_C)^j E_B \Delta E_C (E_B E_C)^{p-j-1} \ket{+}^{\otimes n} + o_r(1).
	\end{equation*}
	
	Each term $\Delta P_j$ in the sum satisfies $\Delta P_j = o_r(n^{-1})$, so the total deviation remains bounded by $o_r(1)$. Consequently, with probability $1 - \mathcal{O}(\mathrm{erfc}(n^{\delta/2}))$, the overall deviation $\Delta P$ is $o(1)$, implying that the universal-mixer variant preserves the efficiency of the ideal $k$-local search problem up to a negligible error. 
\end{proof}

\subsection{Proof of efficiency of adiabatic variant on max-\textit{k}-SSAT}
\label{apsubsec_proof_adiabatic_on_SAT}
% 证明绝热的有效性，并证明average-case complexity

Having established the efficiency of the universal-mixer $k$-local quantum search algorithm on random max-$k$-SSAT instances, we now turn to its adiabatic variant. Building upon the results of the previous subsection, the analysis of the adiabatic case largely follows the discussion in Ref.~\cite{mine1}. For completeness and to aid the reader's understanding, we provide a brief outline of the proof strategy here, highlighting the key steps and the connection to the average-case complexity on random instances.

We consider the standard interpolation between the mixing Hamiltonian and the problem Hamiltonian:
\begin{equation}
	H_k(s) =  s H_k + (1-s) H_{B,1}, \quad s \in [0,1],
\end{equation}
with the initial state $\ket{+}^{\otimes n}$ being the highest-energy eigenstate of $H_{B,1}$. By slowly increasing $s$ from 0 to 1, the system adiabatically follows the highest-energy eigenstate of $H_k(s)$, reaching the target state $\ket{t}$ at $s=1$, which maximizes the objective function of the instance. 

According to the standard adiabatic condition, the total runtime $T$ is determined by the minimum spectral gap $g_{k,0}$ of $H_k(s)$. As shown in Lemma~\ref{lemma_gap}, for the ideal $k$-local Hamiltonian, the gap at $s = 1/2$ scales as $\Theta(n^{-1})$. In cases where the locality of the problem Hamiltonian and the mixing Hamiltonian do not perfectly match, the minimum gap $g_{k,0}$ may shift away from $s=1/2$. However, for small constant $k$, this shift is limited, and the gaps near $s=0$ and $s=1$ also scale as $\Theta(n^{-1})$. Therefore, the overall minimum gap $g_{k,0}$ remains of order $\Theta(n^{-1})$, implying a total runtime $T = \mathcal{O}(n^2)$ for solving the $k$-local search problem.

When applied to random max-$k$-SSAT instances, the $k$-local Hamiltonian $H_k$ is replaced by the normalized problem Hamiltonian $\bar{H}_C$, introducing a perturbation $\Delta H_C = \bar{H}_C - H_k$. Unlike circuit-based implementations, where errors accumulate with each iteration, adiabatic evolution distributes the effect of such perturbations over the entire evolution path. In particular, increasing the total runtime $T$ can effectively suppress the impact of $\Delta H_C$, in contrast to circuit-based quantum search where cumulative errors grow with successive iterations.

This inherent robustness allows the adiabatic variant to maintain efficiency even for smaller clause numbers $m$. The underlying effect is formalized in Ref.~\cite{mine1} and applies equally to the universal-mixer variant. We restate the relevant results here as Lemma~\ref{lemma_AQS_on_SAT} and Theorem~\ref{theo_adiabatic_k_local_QS}. The proof strategy directly follows Ref.~\cite{mine1}, with the only modification being the replacement of the mixer Hamiltonian $H_{B,k}$ by the universal-mixer Hamiltonian; all other steps remain unchanged.

\begin{lem}[Lemma~3.2 in Ref.~\cite{mine1}]\label{lemma_AQS_on_SAT}
	Assume the $k$-local quantum search algorithm is efficient on random max-$k$-SSAT instances with $m = \Theta(f(n))$, $\theta = \Theta(n^{-1})$, and $p_\theta = \mathcal{O}(n^2)$. Then the adiabatic $k$-local search achieves comparable efficiency for $m = \Theta(f^{1/2}(n))$ with runtime $\mathcal{O}(n^2)$.
\end{lem}

\begin{thm}[Theorem~3 in Ref.~\cite{mine1}]\label{theo_adiabatic_k_local_QS}
	For any fixed $\epsilon>0$ and sufficiently large $n$, the adiabatic $k$-local quantum search applied to random max-$k$-SSAT instances drawn from $F_f(n,m,k)$ with $m = \Theta(n^{1+\epsilon+\delta})$ achieves efficiency comparable to that of the ideal $k$-local search problem, with probability $1 - \mathcal{O}(\mathrm{erfc}(n^{\delta/2}))$.
\end{thm}

In the previous analysis, we have established the efficiency of the proposed algorithms on random max-$k$-SSAT instances drawn from the distribution $F_f(n,m,k)$. These results naturally carry over to the widely studied distribution $F_s(n,m,k)$, as formalized in Lemma~\ref{lem_model_reduction}.

Furthermore, based on average-case complexity arguments, one can conclude that the algorithm solves random max-$k$-SSAT instances drawn from $F_s(n,m,k)$ with $m = \Omega(n^{2+\epsilon})$ in polynomial time on average. Specifically, by treating Grover search as a fallback for instances where no efficiency guarantee exists, the complexity in such rare cases is 
\[
\mathcal{O}(\text{erfc}(n^{1/2+\epsilon/4}) \, 2^{n/2}) = o(1),
\]
which vanishes asymptotically. Consequently, the algorithm achieves polynomial average-case complexity for the random max-$k$-SSAT problem under the $F_s(n,m,k)$ with $m = \Omega(n^{2+\epsilon})$.

\section{Generalized mixer constructions}
\label{apsec_general_mixer}

In the previous sections, we focused on the universal-mixer variant, which provides a discusstion with fully decoupled 1-local mixing Hamiltonian. In this section, we extend the discussion to a broader class of $k$-local quantum search algorithms by considering more general mixer constructions. These generalized mixers are no longer restricted in the locality of the mixing operator, and can be expressed as
\begin{equation}\label{apeq_generalized_mixer_k_local_QS}
	\ket{\psi} = \left( e^{-i\theta H_{B,k_0}} e^{-i\theta H_{k}} \right)^p \ket{+}^{\otimes n},
\end{equation}
where $k$ and $k_0$ are small constants. The original $k$-local quantum search corresponds to the special case $k = k_0$, while the universal-mixer construction corresponds to $k_0 = 1$. In this section, we show that the generalized-mixer variant retains efficiency comparable to the original $k$-local quantum search, effectively representing a more general formulation of the $k$-local quantum search framework.

As shown in Sec.~\ref{apsec_proof_universal_mixer}, the core of the efficiency proof lies in establishing that the circuit in Eq.~\eqref{apeq_generalized_mixer_k_local_QS} remains effective for the $k$-local search problem. Following the approach in Sec.~\ref{apsubsec_proof_search_on_search}, the key step is to extend the proofs of Lemmas~\ref{lemma_eigen_reduction} and \ref{lemma_gap} to the generalized-mixer variant. To avoid unnecessary repetition, we present here only the statements and proofs of these lemmas in the generalized-mixer setting as Lemmas~\ref{lemma_generalized_mixer_eigen_reduction} and \ref{lemma_generalized_mixer_gap}, omitting the remaining derivations.

\begin{lem}\label{lemma_generalized_mixer_eigen_reduction}
	Denotes 	
	\begin{align}\label{eq_universal_mixer_approxi_Hn}
		\tilde{\mathcal{H}}_{n+1,k,k_0} 
		& = \mathrm{diag}\{ s(k) H_{n,k} + s(k_0) H_{B,n,k_0}, \nonumber\\
		& \quad\,\, s(k) H_{n,k} + s(k_0) H_{B,n,k_0} + \nonumber\\
		& \quad\,\, \bar{s}(k) H_{n,k-1} + \bar{s}(k_0) H_{B,n,k_0-1}
		\},
	\end{align}
	where $s(k) = \frac{n+1-k}{n+1}$ and $\bar{s}(k) = 1 - s(k)$. For sufficiently large $n$, the eigenvalues of $\mathcal{H}_{n+1,k,k_0} = H_{B,n,k_0} + H_{n,k}$ are of the same order of magnitude as those of $\tilde{\mathcal{H}}_{n+1,k,k_0}$. 
\end{lem}
\begin{proof}	
	Initially, the case $k = k_0$ has been analyzed in Ref.~\cite{mine1} for the original $k$-local quantum search, while the case $k_0 = 1$ was addressed in the previous section. By symmetry between $k$ and $k_0$, it therefore suffices to consider the case $k > k_0 > 1$. Similar to Lemma~\ref{eq_approxi_Hn}, we set $\theta = \Theta(n^{-1})$, and the analysis of the spectral gap of $\mathcal{H}_{n,k,k_0}$ can be reduced to the study of the Hamiltonian
	\begin{equation*}
		U_{n+1,k,k_0}(\theta) = e^{-i\theta H_{B,n+1,k_0}} e^{-i\theta H_{n+1,k}}.
	\end{equation*}
	
	According to whether the ($n$+1)-th qubit participates in the evolution, the Hamiltonian $H_{n+1,k}$ can be decomposed as
	\begin{equation*}
		H_{n+1,k} = s(k) I \otimes H_{n,k} + \bar{s}(k) H'_{n,k-1},
	\end{equation*}
	where $H'_{n,k-1}$ corresponds to $H_{n,k-1}$ with an additional control qubit at the ($n$+1)-th position. Based on this decomposition, the unitary $U_{n+1,k,k_0}(\theta)$ can be expanded as	
	\begin{align*}
		U_{n+1,k,k_0}(\theta) 
		& = e^{-i\theta s(k_0) H_{B,n,k_0}} \, e^{-i\theta \bar{s}(k_0) H'_{B,n,k_0-1}} \nonumber\\
		& \quad \times e^{-i\theta \bar{s}(k) H'_{n,k-1}} \, e^{-i\theta s(k) H_{n,k}}.
	\end{align*}
	
	The operators $e^{-i\theta \bar{s}(k_0) H'_{B,n,k_0-1}}$ and $e^{-i\theta \bar{s}(k) H'_{n,k-1}}$ approximately commute, with a deviation of order ${\mathcal O}(\theta^2 n^{-2})$, as established in Lemma~\ref{lemma_commutator}. Consequently, the evolution can be written as
	 \begin{align*}
	 	U_{n+1,k,k_0}(\theta) & = e^{-i\theta s(k_0) H_{B,n,k_0}} e^{-i\theta \bar{s}(k) H'_{n,k-1}} \\\nonumber
	 	& \quad \times e^{-i\theta \bar{s}(k_0) H'_{B,n,k_0-1}} e^{-i\theta s(k) H_{n,k}} + {\mathcal O}(n^{-4}).
	 \end{align*}
	 Neglecting this small deviation, we denote the main component of $U_{n+1,k,k_0}(\theta)$ as $U_m(\theta)$. It can be decomposed into two subcomponents,
	 \begin{equation*}
	 	U_m(\theta) = U_{m1}(\theta) \, U_{m2}(\theta),
	 \end{equation*}
	 with
	 \begin{align*}
	 	& U_{m1}(\theta) = H^{\otimes n+1} e^{-i\theta s(k_0) H_{n,k_0} }H^{\otimes n+1} e^{- i\theta \bar{s}(k) H'_{n,k-1}}, \nonumber\\
	 	&  U_{m2}(\theta) = H^{\otimes n+1} e^{-i\theta \bar{s}(k_0) H'_{n,k_0-1}} H^{\otimes n+1} e^{-i\theta s{k} H_{n,k}}.
	 \end{align*}
	 
	For an arbitrary computational basis $\ket{x}$, the evolution of $U_{m1}(\theta)$ on the states $\ket{0} \ket{x}$ and $\ket{1} \ket{x}$ are expressed as 
	 \begin{align*}
	 	& U_{m1}(\theta) \ket{0} \ket{x} = e^{-i\theta s(k_0) H_{B, n,k_0} } \ket{0} \ket{x}, \nonumber\\
	 	& U_{m1}(\theta) \ket{1} \ket{x} = e^{-i\theta s(k_0) H_{B, n,k_0} } e^{- i\theta \bar{s}(k) H_{n,k-1}}  \ket{1} \ket{x}. 
	 \end{align*}
	 Thus, the evolution operator $U_{m1}(\theta)$ can be reformulated in block-matrix form as
	 \begin{align*}
	 	& U_{m1}(\theta) = \\\nonumber
	 	&\quad\quad \left[ \begin{matrix}
	 		e^{-i\theta s(k_0) H_{B, n,k_0} } &  \\
	 		& e^{-i\theta s(k_0) H_{B, n,k_0} } e^{-i\theta \bar{s}(k) H_{n,k-1}}  \\
	 	\end{matrix} \right]. 
	 \end{align*}
	 
	 Similarly, the evolution of $H^{(n+1)} U_{m2}(\theta)H^{(n+1)}$, where $H^{(n+1)}$ denotes the Hadamard gate on the ($n$+1)-th qubit, can be expressed as 
	 \begin{align*}
	 	H^{(n+1)} U_{m2}(\theta) & H^{(n+1)} \ket{0} \ket{x} = e^{-i\theta s(k) H_{n,k} } \ket{0} \ket{x} ,\nonumber\\
	 	H^{(n+1)} U_{m2}(\theta) & H^{(n+1)} \ket{1} \ket{x} = \nonumber\\
	 	& e^{- i\theta \bar{s}(k_0) H_{B, n,k_0-1}} e^{-i\theta s(k) H_{n,k} } \ket{1} \ket{x}.
	 \end{align*}
	 Equivalently, this operator can be represented in block-matrix form as
	 \begin{align*}
	 	H^{(n+1)} U_{m2} & (\theta)H^{(n+1)} = \nonumber\\
	 	& \left[ \begin{matrix}
	 		e^{-i\theta s(k) H_{n,k_1} } &  \\
	 		& e^{-i\theta \bar{s}(k_0) H_{B, n,k_0-1}} e^{-i\theta s(k) H_{n,k} } \\
	 	\end{matrix} \right]. 
	 \end{align*}

	Similar to the previous analysis, we can ignore the extra pair of Hadamard gates $H^{(n+1)}$ when analyzing a large-scale quantum system, and focus on the eigendecomposition of $U'_m = U_{m1}(\theta)H^{(n+1)}U_{m2}(\theta)H^{(n+1)}$, which can be represented in block-matrix form as 
	\begin{equation*}
		U'_a = \left[ \begin{matrix}
			U_{m00}  &  \\
			& U_{m11}  \\
		\end{matrix} \right],
	\end{equation*}
	where 
	\begin{align*}
		U_{m00} & = e^{-i\theta s(k_0) H_{B, n,k_0} } e^{-i\theta s(k) H_{n,k}},  \nonumber\\
		U_{m11} & = e^{-i\theta s(k_0)  H_{B, n,k_0} } e^{-i\theta \bar{s}(k)  H_{n,k-1}} \nonumber\\
		& \quad \times e^{-i \theta \bar{s}(k_0) H_{B, n,k_0-1}}  e^{-i\theta s(k) H_{n,k}}. 
	\end{align*}
	According to the Trotter-Suzuki decomposition, $U'_m$ has an eigenspectrum identical to that of $e^{-i\theta \tilde{\mathcal{H}}_{n+1,k,k_0}}$, and the lemma follows.
\end{proof}

\begin{lem}\label{lemma_generalized_mixer_gap} 
	For a small constant $k$, the spectral gap $g_{k,k_0}$ of $\mathcal{H}_{n,k,k_0}$ scales as $\Theta(n^{-1})$. 
\end{lem}
\begin{proof}
	Similar to Lemma~\ref{lemma_generalized_mixer_eigen_reduction}, it suffices to focus on the case of small constants $k > k_0 > 1$. Assume, as the induction hypothesis, that the lemma holds for the case $(k-1, k_0-1)$ with any $n$ sufficiently larger than $k$, and also for the case $(k, k_0)$ with a specific $n$ moderately larger than $k$.  We now aim to establish the lemma for the case $(k, k_0)$ with $n+1$.
	
	% Proof of Lemma~\ref{lemma_generalized_mixer_gap}
	According to Lemma~\ref{lemma_generalized_mixer_eigen_reduction}, the spectral gap of $\mathcal{H}_{n+1,k,k_0}$ can be analyzed by examining each block of $\tilde{\mathcal{H}}_{n+1,k,k_0}$, as defined in Eq.~\eqref{eq_approxi_Hn}. Specifically, the gap is determined both by the gaps within the individual blocks and by the difference between the largest eigenvalues of the two blocks.
	
	For the upper block $H_{m00} = s(k) H_{n,k} + s(k_0) H_{B,n,k_0}$, since $k$ is a small constant, the spectral gap of this block is of the same order as that of $\mathcal{H}_{n,k,k_0}$, namely $\Theta(n^{-1})$.  For the lower block $
	H_{m11} = s(k) H_{n,k} + s(k_0) H_{B,n,k_0} + \bar{s}(k) H_{n,k-1} + \bar{s}(k_0) H_{B,n,k_0-1}$, recalling the deviation of the objective function in Eq.~\eqref{eq_deviation_fkx} and neglecting higher-order terms, the spectral gap of this block also remains of the same order as that of $\mathcal{H}_{n,k,k_0}$.
	
	Regarding the gap between the two blocks, it is primarily determined by $
	\Delta H_m = \bar{s}(k) H_{n,k-1} + \bar{s}(k_0) H_{B,n,k_0-1}$, whose magnitude is also clearly $\Theta(n^{-1})$. Combining these observations, we conclude that the spectral gap $g_{k,k_0}$ of $\mathcal{H}_{n+1,k,k_0}$ scales as $\Theta(n^{-1})$, proving the lemma.
\end{proof}

Based on the two key lemmas above, and following a reasoning similar to that in Sec.~\ref{apsec_proof_universal_mixer}, we can establish the efficiency of the generalized-mixer $k$-local quantum search and its adiabatic variant on random max-$k$-SSAT instances. Having established a unified framework for $k$-local quantum search, we now formulate a theorem to characterize its performance. For clarity, hereafter the term ``$k$-local quantum search'' refers specifically to the generalized-mixer variant, while the original formulation (where $k = k_0$) will be distinguished as the ``symmetric $k$-local quantum search''.

\begin{thm}\label{theo_main}
	Given a random instance of max-$k$-SSAT drawn from $F_s(n,m,k)$ with $m = \Theta(n^{1+\epsilon+\delta})$, for any fixed $\epsilon>0$ and sufficiently large $n$, the adiabatic $k$-local quantum search achieves efficiency comparable to that of the ideal $k$-local search problem, with probability $1 - \mathcal{O}(\mathrm{erfc}(n^{\delta/2}))$.
\end{thm}

\section{Analysis of adiabatic-manifold-based reparameterization}
\label{apsec_analysis_am_repara}

This subsection focuses on how the smooth adiabatic manifold can be exploited to construct a more effective reparameterization of the variational parameter space. Since the adiabatic manifold is closely related to the underlying adiabatic evolution trajectory, it is natural to generalize the linear schedule in Eq.~\eqref{eq_paraform_two} according to the analysis of the optimal adiabatic trajectory presented in Sec.~\ref{subsec_connecting_qaoa}. Instead of adopting fixed linear coefficients, we allow the interpolation between the two Hamiltonians to be governed by continuous functions of time, yielding
\begin{equation}
	H(t) = f_C(s) \bar{H}_C + f_B H_{B,1},
\end{equation}
subject to the boundary conditions $f_C(0) = 0, \, f_B(0) > 0 \, f_C(T) > 0 ,\ f_B(T) = 0$ with $0\le t\le T$. As the evolution proceeds, the system Hamiltonian evolves smoothly from the mixer Hamiltonian toward the problem Hamiltonian. The pair of functions $\bm{f}(t)=(f_C(t),f_B(t))$ therefore completely characterizes the adiabatic trajectory.

Importantly, the optimal adiabatic trajectory $\bm{f}^*(t)$ is expected to remain continuous with respect to time. Intuitively, if an ``optimal'' schedule contained discontinuities, one could always construct a smoother alternative through local interpolation, thereby reducing the magnitude of $\frac{{\rm d}H(t)}{{\rm d}t}$, which in turn relaxes the adiabatic constraint and decreases the required evolution time. Consequently, the effective variational manifold induced by the compressed adiabatic evolution is naturally expected to inherit a smooth geometric structure.

To simulate the optimal evolution, the optimal adiabatic passage $\bm{f}^*(t)$ is discretized into $p$ intervals. For the $d$-th interval $\left(t_{d-1}, t_d\right]$, the time evolution is approximated by
\begin{equation}\label{eq_AP_para}
	e^{iH(t_d)\Delta t} \approx e^{if^*_C(t_d)\Delta t_d H_C}e^{if^*_B(t_d)\Delta t_d H_B}, 
\end{equation}
where $\Delta t_d = t_d - t_{d-1}$. In QAOA, the total evolution is parameterized by introducing variational parameters $\gamma_d=f^*_C(t_d)\Delta t_d$ and $\beta_d=f^*_B(t_d)\Delta t_d$, effectively encoding the adiabatic passage into the sequence of alternating unitaries. If no information is available about the underlying structure of $\bm{f}^*(t)$, this direct parameterization serves as a natural and efficient approach. However, since $\bm{f}^*(t)$ is known to be continuous, there exists an opportunity to exploit this property for a more structured and potentially more efficient parameterization of $\bm{f}(t)$.

% give the relation of the factor on continuity
In the evolution described by Eq.~(\ref{eq_AP_para}), three explicit factors appear: the time duration $\Delta t_d$, and the coefficients $f^*_C(t_d)$ and $f^*_B(t_d)$. The interdependence among them is governed by the optimality conditions of the adiabatic passage functions $f^*_C(t)$ and $f^*_B(t)$. According to the adiabatic condition, the time step $\Delta t_d$ should satisfy
\begin{equation}\label{eq_factor_relation}
	\Delta t_d\sim \varepsilon(t_d)g^{-2}(t_d).
\end{equation}
where $\varepsilon(t_d)$ characterizes the instantaneous rate of change of the Hamiltonian. Based on Eq.~\eqref{eq_T_varepsilon},$\varepsilon(t_d)$ is primarily determined by the derivative ${\rm d}H(t)/{\rm d}t$. Thus, it is not the absolute values of $(f^*_C(t_d), f^*_B(t_d))$ that govern $\Delta t_d$, but rather their local variations. For convenience, we denote
\begin{equation}
	\begin{aligned}
		\Delta f^*_C(t_d) &= f^*_C(t_d) - f^*_C(t_{d-1}), \\
		\Delta f^*_B(t_d) &= f^*_B(t_{d-1}) - f^*_B(t_d).
	\end{aligned}
\end{equation}
Accordingly, the evolution at each step depends on four tightly coupled quantities: the time duration $\Delta t_d$, the spectral gap $g(t_d)$, and the variations $\Delta f^*_C(t_d)$ and $\Delta f^*_B(t_d)$. These quantities together determine the necessary conditions for preserving adiabaticity throughout the discretized evolution.

% analysis the continuity of each factor
Because the gap $g(t)$ is an intrinsic property of the Hamiltonian that cannot be directly adjusted, it can be treated as an independent variable. The other three factors---$\Delta t_d$, $\Delta f^*_C(t_d)$, and $\Delta f^*_B(t_d)$---depend on $g(t_d)$ and exhibit varying degrees of continuity. While discrete variables lack strict continuity, we refer here to approximate continuity, meaning that the variations between adjacent intervals remain small.

According to Eq.~\eqref{eq_factor_relation}, the variations  $\Delta f^*_C(t_d)$ and $\Delta f^*_B(t_d)$ tend to be positively correlated with $g(t_d)$, while the time step $\Delta t_d$ is inversely proportional to $g(t_d).$ However, these factors are also subject to Trotterization constraints, which limit their magnitudes. As a result, when $g(t_d)$ becomes small, $\Delta f^*_C(t_d)$ and $\Delta f^*_B(t_d)$ are typically reduced, preventing a compensatory quadratic increase in $\Delta t_d$.

%%%% another explain %%%%
% parameterization of better continuity, intro  (theta, rho)
The continuity of parameters plays a crucial role in setting effective schedules. Accordingly, the parameterization of the adiabatic passage should preserve the best possible continuity among the relevant factors. Notably, the incremental form $\Delta f^*_C(t_d)$ exhibits better continuity than the cumulative quantity $f^*_C(t_d) = \sum_{j=1}^d \Delta f^*_C(t_j)$. Moreover, based on the analysis of the normalized Hamiltonian, the variations $\Delta f^*_C(t_d)$ and $\Delta f^*_C(t_d)$ are expected to share similar behaviors with the parameters $f_\gamma$ and $f_\beta$. This observation motivates the mapping
\begin{equation}
	(\Delta f^*_C(t_d), \Delta f^*_C(t_d)) \mapsto (\rho_d\sin{\theta_d}, \rho_d\cos{\theta_d}),
\end{equation}
where $\rho_d$ represents the magnitude of the Hamiltonian’s rate of change at step $d$, and $\theta_d$ captures the relative contributions of $H_C$ and $H_B$ in that interval. A natural approach is to parameterize these dependent factors directly, leading to the following mapping for the QAOA parameters:
\begin{equation}\label{eq_para3}
	\gamma_d=\sum_{j=1}^d \rho_j\sin{\theta_j}\Delta t_d,
	\beta_d=\sum_{j=d}^p \rho_j\cos{\theta_j}\Delta t_d. 
\end{equation}
Here, $\theta_d$ encodes the ratio between the increments of $H_C$ and $H_B$ in the $d$-th interval. Under proper normalization, its expected value is approximately $\pi/4$, and it tends to vary smoothly across steps. In contrast, $\rho_d$ captures the total strength of the Hamiltonian’s evolution and is positively correlated with the gap $g(t_d)$. Compared to individual components like $\Delta f^*_C(t_d)$ or $\Delta f^*_C(t_d)$, typically exhibits better continuity.

This ternary parameterization offers favorable continuity and interpretability. However, the introduction of $3p$ parameters significantly increases the optimization cost. To address this, a reduction in parameter count can be achieved by approximating the cumulative behavior of the summation terms. For instance, considering the component associated with $H_C$, we define
\begin{equation}
	\kappa_d = \frac{1}{p\rho_d\sin{\theta_d}}\sum_{j=1}^{d}\rho_j\sin{\theta_j}. 
\end{equation} 
Here, each term $\rho_j\sin{\theta_j}$ can be approximately treated as an i.i.d. random variable. As $d$ increases, the cumulative sum tends to exhibit a dominant linear growth, allowing the approximation $\kappa_d =\frac{d}{p}+\Delta\kappa_d$, where $\Delta\kappa_d$ captures the residual fluctuations around the linear trend.

% The remaining factors in Eq.~(\ref{eq_para3}) are $p$, $\rho_d$, $\Delta t_d$, and $\sin{\theta_d}$. 
In Eq.~(\ref{eq_para3}), $\gamma_d$ can be rewritten as $p \kappa_d \rho_d \Delta t_d \sin{\theta_d}$. Since $\rho_d$ and $\Delta t_d$ exhibit opposite correlations with the gap $g(t_d)$, their product $\rho_d \Delta t_d$ tends to vary more smoothly and therefore exhibits better continuity. Additionally, the total evolution time satisfies $\sum_{d=1}^p \Delta t_d = T$, implying that each $\Delta t_d$ is approximately inversely proportional to $p$. As a result, the term $p \rho_d \Delta t_d$ can be absorbed into a single parameter $\tau_d$ with improved continuity.

Regarding $\kappa_d$, the lower-order fluctuation term $\Delta \kappa_d$ can be further decomposed into two components: an amplitude-related fluctuation that can be absorbed into $\tau_d$, and a phase-related fluctuation that can be absorbed into $\theta_d$. Therefore, the original three-parameter system can be reduced to a two-parameter one, denoted by $({\bm\theta, \bm\tau})$. Their correspondence to the original parameterization $({\bm \gamma}, {\bm \beta})$ is given by
\begin{equation}
	\gamma_d=\frac{d}{p+1}\tau_d\sin{\theta_d}, \beta_d=\frac{p+1-d}{p+1}\tau_d\cos{\theta_d}.
\end{equation}
fter a straightforward parameter substitution and rescaling, this formulation becomes equivalent to the form presented in Eq.~\eqref{eq_repara}.

\section{Experimental Setup}\label{app_exp}

This appendix summarizes the experimental settings used throughout the simulations presented in this paper.%Section~\ref{sec_comp_pref}.

\paragraph{Circuit Depth.}

All simulations were conducted with the QAOA depth fixed at $p = n$, where $n$ denotes the number of variables in each problem instance. This setting provides a balance between practical feasibility and theoretical relevance. On the one hand, superlinear depths are typically impractical for near-term quantum devices; on the other hand, constant or logarithmic depths may fail to capture the intrinsic computational hardness of NP-complete problems. 

Our method remains applicable for depths exceeding $n$, where the optimization cost continues to stay at a relatively low level. For shallower circuits, the success probability tends to decrease but still remains sufficiently high for meaningful performance. Therefore, we adopt $p = n$ as a representative and well-balanced configuration throughout our experiments.

\paragraph{Problem Instance Generation.}
Random instances of 3-SAT were generated using a satisfiable-instance-restricted distribution $F_s(n, m_n^*, 3)$, where each instance is guaranteed to be satisfiable. The clause-to-variable ratio was chosen as $m_n^* = \mu_n n$, with each $\mu_n$ empirically tuned so that the expected number of satisfying assignments is approximately 1.3. Table~\ref{tab:mu_values} lists the $\mu_n$ values used for $n=4$ to $20$, covering the phase transition region that is widely considered computationally challenging~\cite{Achioptas2006, Kirousis1998, Coja2016}.

\begin{table}[!t]
	\footnotesize
	\caption{\small Empirical values of $\mu_n$ used to generate clause-to-variable ratios $m_n^* = \mu_n n$ for satisfiable 3-SAT instances. These values yield approximately 1.3 expected satisfying assignments per instance.}
	\label{tab:mu_values}
	\tabcolsep 2.75pt
	\begin{tabular*}{0.47\textwidth}{ccccccccc}
		\toprule
		$n$ & 4 & 5 & 6 & 7 & 8 & 9 & 10 & 11 \\\hline
		$\mu_n$ & 5.460 & 5.413 & 5.833 & 5.781 & 5.900 & 5.895 & 6.075 & 6.036 \\
		\bottomrule
		12 & 13 & 14 & 15 & 16 & 17 & 18 & 19 & 20 \\\hline
		6.090 & 6.191 & 6.160 & 6.234 & 6.320 & 6.396 & 6.368 & 6.460 & 6.413 \\
		\bottomrule
	\end{tabular*}
\end{table}

\paragraph{Optimization Procedure.}
For parameter optimization, we employed the BFGS algorithm, consistent with the setup used in~\cite{Zhou2020}, with the objective of maximizing the expectation value $\langle H_C \rangle$. Unlike many prior works that report the number of optimization epochs, we measure optimization cost by the total number of expectation evaluations of the cost Hamiltonian $\langle H_C \rangle$. This metric better reflects the true computational effort, as each gradient-based update typically requires $\mathcal{O}(p)$ circuit executions.

\paragraph{Evaluation Metric.}
We evaluate performance using the \emph{success probability}, defined as the probability of sampling a satisfying assignment from the QAOA output state. This metric is adopted in place of approximation ratios or expected energies to reflect the nature of the problem reduction described in Eq.\eqref{eq_QAOA_opt}, where an exact interpretation of the solution is required to determine satisfiability.

In practice, we find that this metric provides sufficient resolution to distinguish between different parameter-setting strategies across random instances. It ensures consistency across evaluations and aligns with the goal of solving decision-type problems, making it well-suited for both empirical comparison and theoretical analysis in our follow-up work. 

\paragraph{Circuit Execution Assumptions.}
All simulations were conducted in a noiseless setting. To approximate realistic execution costs, each expectation value $\langle H_C \rangle$ is estimated using 1000 measurement shots, following standard practices in variational quantum algorithms. The cost of a single circuit execution is modeled to scale as $\mathcal{O}(mp)$, where $m$ is the number of clauses in the cost Hamiltonian and $p$ is the circuit depth. This scaling reflects the number of two-qubit gates and circuit duration, both of which are key constraints on near-term quantum hardware.

\section{Algorithmic details for baseline methods}

This appendix provides the full pseudocode for three baseline parameter-setting strategies used in our experiments: the FOURIER heuristics introduced in~\cite{Zhou2020} and the TQA-based initialization proposed in~\cite{Sack2021}. These algorithms serve as standard benchmarks for evaluating the performance of our proposed SAMP method. %For completeness and reproducibility, we present each method in algorithmic form below.

\subsection{FOURIER heuristic}\label{app_FOUR}

The FOURIER method starts from a shallow QAOA circuit with randomly initialized parameters. However, its core idea differs in how information is carried across depths: instead of directly propagating optimized parameters, it maintains and updates their underlying frequency components. 

Specifically, QAOA parameters are represented as a truncated Fourier series, that is, using a small number of low-frequency sine and cosine modes. As the circuit depth increases, additional Fourier modes are introduced, gradually enhancing the expressiveness of the parameter schedule. By optimizing in a low-dimensional space of frequency coefficients, the method significantly reduces the complexity of the parameter search. This low-dimensional space is progressively expanded until the full parameter space is effectively explored. The detailed procedure is provided in Algorithm~\ref{alg_fourier_appendix}.

\begin{algorithm}[H]
	\caption{\small FOURIER Heuristic for QAOA}
	\label{alg_fourier_appendix}
	\small
	\begin{algorithmic}[1]
		\State Randomly initialize $\bm{u} = (u_1), \bm{v} = (v_1)$
		\State Set $q \leftarrow 1$
		\While{$q \le p$}
		\State Construct $(\bm{\gamma}, \bm{\beta})$ of length $q$ via Fourier transform:
		\State \quad $\gamma_j = \sum_{k=1}^q u_k \sin\left(\frac{(k-\frac{1}{2})(j-\frac{1}{2})\pi}{q}\right)$
		\State \quad $\beta_j = \sum_{k=1}^q v_k \cos\left(\frac{(k-\frac{1}{2})(j-\frac{1}{2})\pi}{q}\right)$
		\State Optimize $(\bm{u}, \bm{v})$ using QAOA of depth $q$
		\State Set $\bm{u} \leftarrow (u_1, \dots, u_q, 0)$, $\bm{v} \leftarrow (v_1, \dots, v_q, 0)$
		\State $q \leftarrow q + 1$
		\EndWhile
	\end{algorithmic}
\end{algorithm}

\subsection{TQA initialization}\label{app_TQA}

% give the initialization, analyze the method and shows the variation of this method
The TQA-based initialization aims to estimate the statistical optimum $(\bar{\theta}^*, \bar{\rho}^*)$ via instance-based parameter prediction. This strategy proves effective---particularly under our normalized Hamiltonian framework---as it significantly reduces the cost of such prediction by identifying a promising region with favorable gradient directions. However, when the true optimal parameters $(\theta^*, \rho^* )$ significantly deviate from these statistical estimates, it is preferable to avoid overfitting to potentially misleading values and instead initialize with parameters that exhibit stronger gradients.

The Trotterized Quantum Annealing (TQA) method constructs a parameter schedule inspired by linear adiabatic evolution, using precomputed statistical averages of the parameters of QAOA.  This method assumes the parameters in form of
\begin{equation*}
	\gamma_d =  \frac{d\bar{\gamma^*}}{p+1}, \qquad
	\beta_d = \frac{(p-d+1)\bar{\beta^*}}{p+1} , 
\end{equation*}
where $(\bar{\gamma^*}, \bar{\beta^*})$ are statistical averages derived from small-instance optimization.

The TQA method is used solely for parameter initialization. The subsequent optimization is performed directly using classical optimizers, without applying any further heuristics. Accordingly, in Section~\ref{subsec_benchmarking}, a full parameter optimization is conducted to reach a quasi-optimal setting.

\section{Supplementary experiments}

\subsection{Additional results for adiabatic manifold}
\label{apsubsec_extra_simulation}

we conduct additional simulations on a broader set of random instances generated from $F_s(12, m, 3)$, as presented in Figure \ref{fig_exp4g}. These results visualize the probability of measuring the target (satisfying) state as a function of $f_\gamma$ and $f_\beta$, with each probability estimated from 10{,}000 measurement shots and represented via the corresponding raw count.

While the exact shape of the optimal region may vary across instances, a consistent pattern emerges: along the axis $f_\gamma = f_\beta$ (i.e., $\theta = \pi/4$), the probability gradient is generally steep, indicating a robust path toward optimality. This trend persists until the scale of $f_\gamma$ and $f_\beta$ becomes excessively large, beyond which performance rapidly declines and the landscape flattens. These results are consistent with our theoretical explanation and reinforce the effectiveness of the proposed initialization approach.

\begin{figure}[!t]
	\centering
	\footnotesize
	{\includegraphics[width=0.9\linewidth]{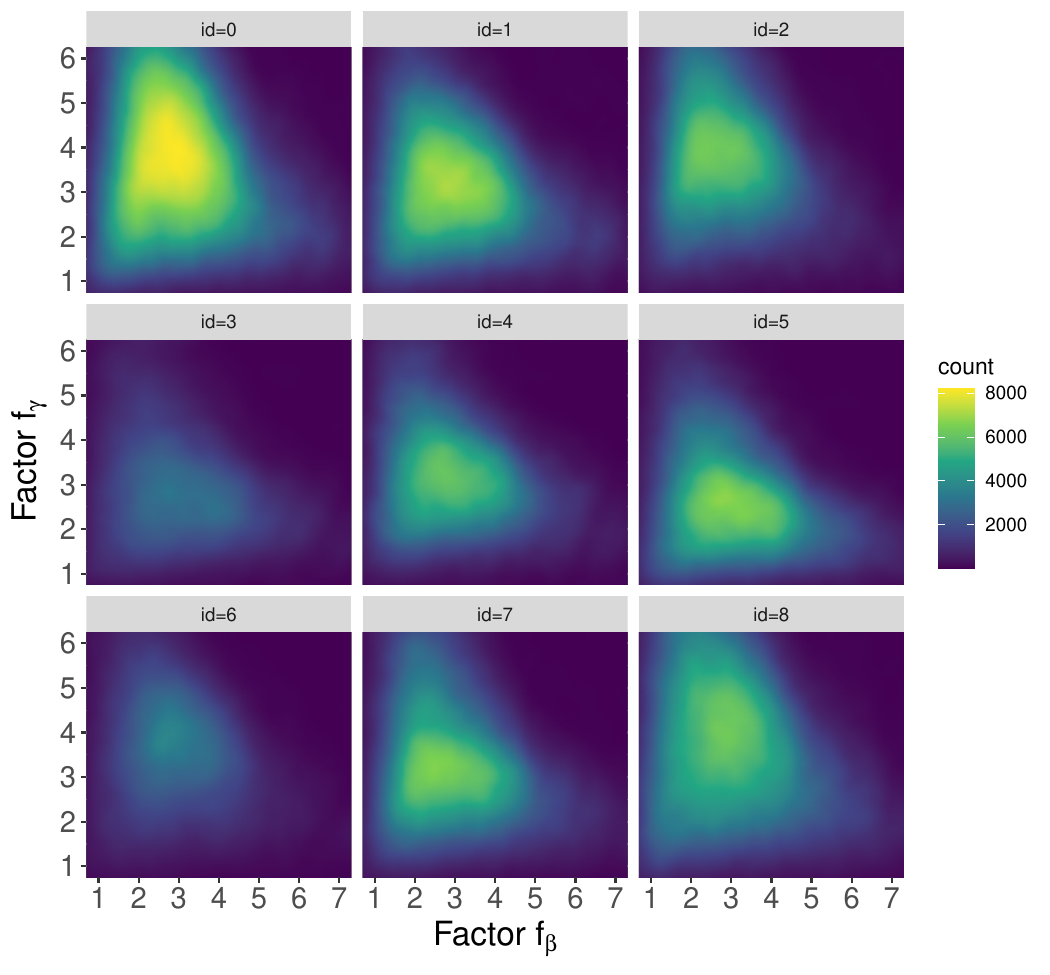}}
	\caption{\small Additional simulations on random instances from $F_s(12, m, 3)$, showing the distribution of measurement counts for obtaining the target state across different values of $f_\gamma$ and $f_\beta$.}
	\label{fig_exp4g}
\end{figure}

\subsection{Experiment on max-cut}\label{app_cut}

In this supplementary simulation, both parameter-setting algorithms proposed in this paper are applied to instances of the Max-Cut problem on random graphs generated using the Erd\H{o}s--R\'enyi model. Each random graph has an edge density of $0.5$, meaning that for a graph with $n$ nodes, the number of edges is approximately $0.5C_n^2$.

The simulation results are presented in Fig.~\ref{fig_cut}. In Fig.~\ref{fig_cutProb}, the significantly high success probabilities demonstrate the efficiency of the SAMP method in solving the Max-Cut problem. Fig.~\ref{fig_cutCount} shows the distribution of optimization costs. The observed trends are consistent with those reported for the 3-SAT problem that the SAMP method exhibits a logarithmic growth in cost as the problem size increases.

\begin{figure}[t]
	\centering
	\begin{subfigure}[t]{0.8\linewidth}
		\centering
		\includegraphics[width=\linewidth]{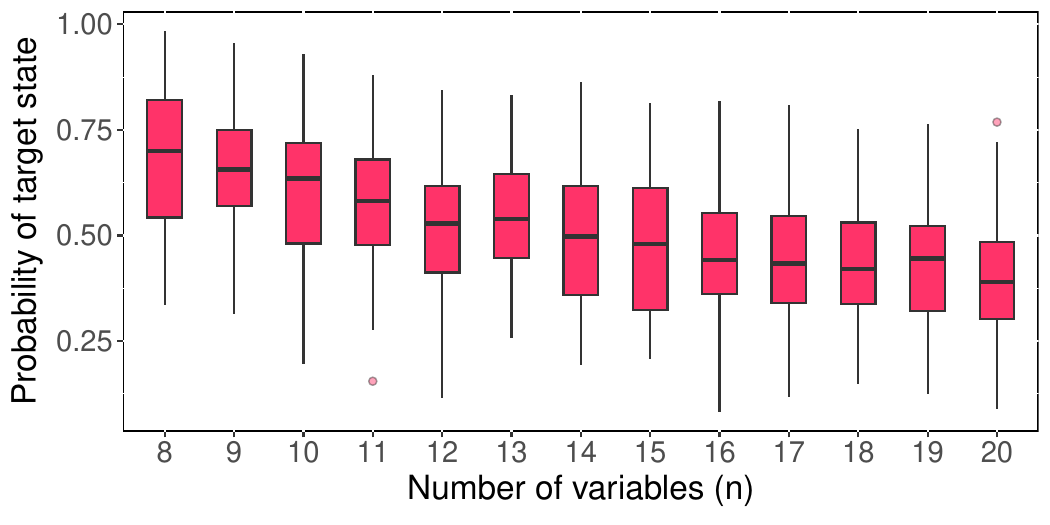}
		\caption{\small Distribution of success probabilities.}
		\label{fig_cutProb}
	\end{subfigure}
	%\vspace{-1.2ex} 
	\begin{subfigure}[t]{0.8\linewidth}
		\centering
		\includegraphics[width=\linewidth]{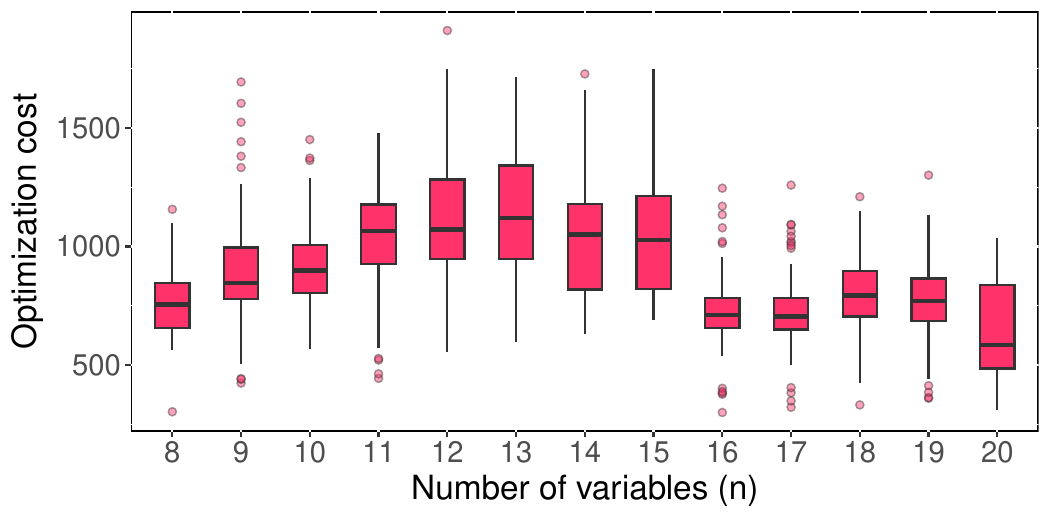}
		\caption{\small Distribution of optimization costs.}
		\label{fig_cutCount}
	\end{subfigure}
	\vspace{-0.6ex}
	\caption{\small Performance of the SAMP method (red, right) on 100 random Erd\H{o}s--R\'enyi graph instances with edge density $0.5$ and node sizes ranging from $n = 8$ to $20$.}
	\label{fig_cut}
\end{figure}

\bibliography{SAMP_aps}% Produces the bibliography via BibTeX.
% 参考文献至少要加到50.

\end{document}